\RequirePackage[2018-12-01]{latexrelease}
\documentclass[printer]{rQUF2e}

\usepackage{mymacros}

\usepackage{epstopdf}
\usepackage{subfigure}
\usepackage{mathrsfs}
\usepackage{mathtools}
\usepackage{multirow}
\usepackage{pgfplots}
\usepackage{siunitx}
\usepackage{tikzexternal}

\pgfplotsset{compat=1.16}
\pgfplotsset{
    every axis/.append style={
        table/col sep=comma,
        tick align=outside,
        tick pos=left,
        label style={font=\small},
        ticklabel style={font=\footnotesize},
        grid style={line width=.1pt, draw=gray!20},
        tick scale binop=\times,
        xmajorgrids=false,
        xticklabel style={rotate=0, anchor=near xticklabel},
        ymajorgrids=false,
        legend style={
            legend cell align=left,
            font=\scriptsize,
        }
    },
    AgentStyle/.style={
        mark options={solid, scale=0.5},
        mark=*,
        nodes near coords,
        point meta=explicit symbolic,
        every node near coord/.append style={font=\scriptsize,anchor=\myanchor},
        txtBlue
    },
    BaselineStyle/.style={
        only marks,
        mark options={solid, scale=0.5},
        mark=*,
        nodes near coords,
        point meta=explicit symbolic,
        every node near coord/.append style={font=\scriptsize,anchor=\myanchor},
        black
    },
    LineStyle/.style={no markers, thin}
}
\DeclareMathOperator*{\EV}{\mathbb{E}}
\DeclareMathOperator*{\EVp}{\mathbb{E}_\pi}
\DeclareMathOperator*{\esssup}{ess\,sup}

\newcommand{\bbE}{\mathbb{E}}
\newcommand{\bbP}{\mathbb{P}}
\newcommand{\bbQ}{\mathbb{Q}}
\def\eps{\varepsilon}
\def\LGD{\ensuremath{\mathsf{LGD}}}
\def\Rec{\ensuremath{\mathsf{Rec}}}
\def\CVA{\ensuremath{\mathsf{CVA}}}
\def\CDS{\mathsf{CDS}}
\def\EUR{\mathsf{EUR}}
\def\USD{\mathsf{USD}}

\newcommand{\Aspace}{\mathcal{A}}
\newcommand{\Sspace}{\mathcal{S}}
\newcommand{\Tspace}{\mathcal{T}}

\theoremstyle{plain}
\newtheorem{theorem}{Theorem}[section]

\newtheorem{lemma}[theorem]{Lemma}

\theoremstyle{definition}
\newtheorem{definition}{Definition}

\theoremstyle{remark}
\newtheorem{remark}{Remark}

\begin{document}

\title{CVA Hedging by Risk-Averse Stochastic-Horizon Reinforcement Learning}

\author{R. DALUISO$^a{\dag}$\thanks{$^a$Corresponding author. Email: roberto.daluiso@intesasanpaolo.com. \orcidlink{0000-0001-5943-5808} \href{https://orcid.org/0000-0001-5943-5808}{https://orcid.org/0000-0001-5943-5808}},
	M. PINCIROLI$^b{\ddag}$\thanks{$^b$\orcidlink{0009-0000-8278-0212} \href{https://orcid.org/0009-0000-8278-0212}{https://orcid.org/0009-0000-8278-0212}},
    M. TRAPLETTI$^c{\ddag}$\thanks{$^c$\orcidlink{0009-0001-1241-4388} \href{https://orcid.org/0009-0001-1241-4388}{https://orcid.org/0009-0001-1241-4388}},
    \and E. VITTORI$^d{\ddag}$\thanks{$^d$\orcidlink{0000-0003-4648-1797} \href{https://orcid.org/0000-0003-4648-1797}{https://orcid.org/0000-0003-4648-1797}}\\
	\affil{$\dag$Interest Rates and Credit Models, IMI Corporate and Investment Banking,\\Intesa Sanpaolo, Largo Mattioli 3, 20121 Milano, Italy\\
	$\ddag$CVA Management and A.I.~Investments, IMI Corporate and Investment Banking,\\Intesa Sanpaolo, Largo Mattioli 3, 20121 Milano, Italy}
	\received{v1.0 released December 2023}
}

\maketitle

\begin{abstract}
This work studies the dynamic risk management of the risk-neutral value of the potential credit losses on a portfolio of derivatives. Sensitivities-based hedging of such liability is sub-optimal because of bid-ask costs, pricing models which cannot be completely realistic, and a discontinuity at default time. We leverage recent advances on risk-averse Reinforcement Learning developed specifically for option hedging with an \textit{ad hoc} practice-aligned objective function aware of pathwise volatility, generalizing them to stochastic horizons. We formalize accurately the evolution of the hedger's portfolio stressing such aspects. We showcase the efficacy of our approach by a numerical study for a portfolio composed of a single FX forward contract.
\end{abstract}


\begin{classcode}C61; C63\end{classcode}

\section{Introduction}

Each time a financial institution has derivatives in place with some counterparty, it must account for the fair value of the losses it will incur if the counterparty defaults. Such fair value is called Credit Valuation Adjustment (CVA), and its variations impact the Profit and Losses of the financial institution: in this respect, it acts just like a high-dimensional hybrid claim, but one much more complex than the single positions, because it depends jointly on the risk-neutral default probabilities of the counterparty and on the underlying risk factors of all the derivatives.

To offset such variations at least partially, specialized traders can use both instruments sensitive to the underlying risk drivers and Credit Default Swaps on the counterparty. The traded amounts could be defined to cancel the first order sensitivities of CVA to market variables, but this standard practitioner's ``delta hedge'' may be highly suboptimal because textbook hypotheses are particularly far from reality:
\begin{enumerate}
	\item Hedging instruments are often expensive to trade (credit spreads and implied volatilities being always among the risk factors);
	\item The dynamics of CVA cannot be continuous at least at default time where the exposure suddenly jumps to the full value of the loss;
	\item The pricing model for a complex object like CVA is most probably misspecified.
\end{enumerate}

For these reasons, we have proposed in \citet{daluiso2023Acva} to give up closed form solutions in favour of a fully flexible simulated environment where an artificial agent can learn a good hedging strategy by interaction and direct optimization. In the test cases presented in the mentioned conference proceedings, the idea performs very well in that:
\begin{enumerate}
	\item It recovers the analytical benchmark when both costs and jump events are neglected;
	\item When transaction costs are non-negligible, it can produce savings without significant increases in Profit and Loss volatility;
	\item When correlations are also present, it is able to save even more by partly cross-hedging the more costly underlyings with the less costly ones;
	\item When the counterparty is very risky, it finds a policy which is quite different from the ``delta hedging'' benchmark, and significantly outperforms it.
\end{enumerate}

Given the above evidence, the aim of this paper is to provide a more complete treatment of CVA hedging in a risk-averse Reinforcement Learning framework. In particular, on the theoretical side we aim to:
\begin{enumerate}
	\item Describe an environment setup which allows a full decoupling of pricing rules from the dynamics of the underlying risk factors.
	\item Give a complete formalization of the hedging problem inclusive of realistic details such as collateral, transaction costs in multiple currencies, and non-negligible interest rates.
	\item Derive a generalization of the Trust-Region Volatility Optimization algorithm working with a stochastic horizon.
	\item Introduce an importance sampling scheme to have a faster and more stable learning of the effect of the rare default event on the objective function.
\end{enumerate}

On the empirical side, we first of all run variants of the experiments in \citet{daluiso2023Acva} to confirm its findings with the generalized and more realistic model of the present paper. Then, we present further tests focused on the algorithmic innovations about stochastic horizon in risk-averse~RL.

The rest of the paper is organized as follows. Section~\ref{sec:literature} surveys relevant previous work. Section~\ref{sec:environment} describes the financial problem in detail while keeping the specific choice of models still open. Section~\ref{sec:RL} contains a brief overview of Reinforcement Learning with risk aversion inclusive of our proposed generalization for stochastic horizons. Section~\ref{sec:impl} specifies the concrete modelling choices we made for the numerical experiments, whose results are presented and commented in section~\ref{sec:numerics}. Section~\ref{sec:conclusion} contains some closing remarks.

\subsection{Related literature}\label{sec:literature}

Our work stands at the crossroads of the classical mathematical finance problem of hedging, financial applications of RL, and risk management of valuation adjustments. Each of such topics taken separately has been the subject of extensive literature, and here we only give a general picture.

On the classical side, hedging of contingent claims has been tackled without machine learning in countless studies. Particularly relevant to our setting are those considering transaction costs or correlation among risk drivers.

As for transaction costs, some authors just postulate the rules to build and dynamically adjust the hedging portfolio, and then concentrate on quantifying the impact on pricing: the earliest attempt in this direction was probably \citet{leland1985option}, followed by \citet{dewynne1994path} and many others up to the recent \citet{burnett2021hedging}. Other authors try to optimize the hedging strategy by stochastic control tools, getting either numerical or asymptotic solutions: e.g.~\citet{davis1993european, hodges1989optimal, whalley1997asymptotic, zakamouline2005optimal} just to mention a few. The objective is usually utility maximization, often defined on terminal wealth, although some authors try definitions based on the local PnL \citep[e.g.][]{kallsen1999utility} which are closer in spirit to our approach.

As for the impact of dependency, the most attention was drawn by the specific ``shadow-delta'' case of co-movements of an asset and its volatility \citep{crepey2004delta, vahamaa2004delta, bartlett2006hedging, alexander2007model, alexander2012regime, hull2017optimal}. A generalization to a generic set of factors with any cardinality was described in \citet{daluiso2017hedging}, which is also unusually close to our research in the choice of the numerical example: indeed, the case study is CVA hedging, although for an interest rate swap.

Turning to machine learning, hedging has been constantly mentioned since the earliest applications to finance \citep{malliaris1993neural, hutchinson1994nonparametric}. However, in such works the focus is typically on learning a pricing function, whose sensitivities are then used to build a traditional ``delta hedge''. In this respect, hedging performance is more a means to the construction of a error metric for the pricing, then a goal in itself. A quite recent review of this stream of literature can be found in \citet{ruf2020neural}.

On the other hand, it was not until recently that machine learning was attempted with hedging as the primary focus.
Again there is a major distinction between those optimizing a risk measure defined on final performance, and those which use a pricing model to define a daily performance. The first family includes the ``deep hedging'' series \citet{buehler2019deep, murray2022deep, mikkila2023empirical}. The present work belongs to the second family and is particularly aligned with the point of view of \citet{vittori2020option}.
It has also a relevant aspect in common with \citet{cao2021deep, cao2023gamma} in that they discuss the possibility to price with a model and simulate with a different one. Further papers which work on the the daily PnL include \citet{kolm2019dynamic, du2020deep, halperin2019qlbs, halperin2020qlbs}. We would also like to mention \citet{ruf2022hedging} who give a modern flavour to the dependence stream mentioned earlier: indeed, they give up the determination of a full dynamic strategy, and focus on the single-step hedging error, trying to improve on it by a machine learning based estimation of the ``shadow-delta''. 

The points of contact of RL with finance go far beyond hedging, touching for instance pricing, asset allocation, and optimal execution. We limit our review to the stream of literature which relates more to our work, concerned with proposals to introduce a trade-off between risk and expected return. This includes pure machine learning papers like \citet{morimura2010nonparametric, tamar2012policy, tamar2013variance, prashanth_actor-critic_2014, garcia15a, tamar2015policy, tamar_sequential_2017, chow2017risk}, and works which are more explicitly finance-oriented like \citet{moody2001learning, shen_risk-averse_2014}. The special measure of risk we warmly advocate for was introduced in \citet{Bisi2020trvo} and adopted in \citet{zhang2020mean, vittori2020option, bisi2021foreign, bisi2022risk, mandelli2023hedging}.

Finally, CVA management has been often discussed in practice-oriented literature \citep[e.g.][]{cesari2009modelling, brigo2013counterparty}, but virtually all published contributions on its hedging have stuck to sensitivities, whose efficient computation is already a challenging task due to the potentially large dimension of the gradient \citep{reghai2015cva, antonov2017pv, gnoatto2021deep, deelstra2022accelerated, giles2023efficient} and even larger size of the Hessian \citep{daluiso2023fast}. Except for the preliminary results on the present research project which were anticipated in \citet{daluiso2023Acva}, we are only aware of attempts to RL based hedging of valuation adjustments in the dissertations \citet{crotti2016reinforcement, vit2018reinforcement, tizzano2018direct, palmisano2019risk, locatelli2021two}, who got more ambiguous results then ours, partly because they tackle the tougher DVA which can be only proxy-hedged, and partly because of their very broad and completely data-driven ambitions.

\section{Financial environment}\label{sec:environment}

\subsection{Mathematical setting}\label{sec:setting}
In this section we set up the notation and mathematical formalism to describe the hedging problem.

The set of assets which can appear in the trading book is denoted by $\mathscr{H}$; in general, it includes the CVA which must be hedged and a finite collection of hedging and funding instruments. We use the symbol $h$ to denote a generic element of $\mathscr{H}$.

The set of currencies which appear in the book is denoted by $\mathscr{C}$. We use the symbol $c$ to denote a generic element of $\mathscr{C}$, while $\bar{c}$ is the main currency used to measure the performance of the trader; it's typically the home currency of his firm, in which the latter expresses its balance sheet. An exchange rate $\phi_t^{c_1c_2}$ represents the fair price in currency $c_2$ of a unit of currency $c_1$ for any $c_1,c_2\in\mathscr{C}$.

The number of units of asset $h$ held at time $t$ is denoted by $N_t^h$; for the purpose of this section it can be any real number, while constraints are added according to the specific application.

We associate to any asset $h$ a book valuation $X_t^h$ at time $t$, expressed in an asset specific currency $c_h$, and a set of dividend processes $(Y_t^{h,c})_{c\in\mathscr{C}}$ representing the cumulated cash flows of $h$ in currency $c$ which are paid up to time $t$ inclusive. Note that both $X_t^h$ and the set of cash flows $(Y_t^{h,c})_{c\in\mathscr{C}}$ can be expressed in the main evaluation currency $\bar{c}$ by the quantities
\begin{align}
	\bar{X}_t^h &:= \phi_t^{c_h\bar{c}}X_t^h,\nonumber \\
	\bar{Y}_t^h &:= \sum_{c\in\mathscr{C}} \int_{t_0}^t \phi_s^{c\bar{c}} \diff{}Y_s^{h,c} = \sum_{c\in\mathscr{C}} \int_{t_0}^t \left(\phi_{s-}^{c\bar{c}} \diff{}Y_s^{h,c} + \diff{}[\phi^{c\bar{c}}, Y^{h,c}]_s \right).\label{eq:conversion_flows}
\end{align}
This allows a trivial definition for the book value
\begin{equation*}\label{eq:value}
	V_t := \sum_{h\in\mathscr{H}} N^h_t \bar{X}_t^h.
\end{equation*}

Finally, we allow for costs due to rebalancing of the book composition $\vc{N}_t := (N_t^h)_{h\in\mathscr{H}}$, e.g.~induced by trading an hedging instrument at a price which does not equal its theoretical book valuation. In this section we can stay fairly general and just define $T_t(\vc{N})$ as the cumulative costs of trading strategy $\vc{N}$ up to time $t$ exclusive, expressed in the evaluation currency $\bar{c}$; needless to say, in practical applications one needs a concrete model defining $T_t(\vc{N})$ as a function of the paths of market drivers. See section~\ref{sec:costs} for a possible simple choice.

Given the above definition, we can eventually define the cumulative gains of the book from an initial time $t_0$ as
\begin{equation}\label{eq:gain}
	G_t := \sum_{h\in\mathscr{H}}\int_{t_0}^t N_s^h \diff{}\bar{X}_s^h + \sum_{h\in\mathscr{H}} \int_{t_0}^t N_s^h \diff{}\bar{Y}_s^h - T_t(\vc{N}).
\end{equation}

\begin{remark}[Stochastic integration]
We deliberately overlooked mathematical technicalities on the definition of the stochastic integrals which appear in the above formulas and later. One can safely imagine that everything is defined on a filtered probability space satisfying the usual conditions, that the market processes $X^h$, $Y^{h,c}$, and $\phi^{c_1c_2}$ are all semi-martingales, and that $N^h$ is taken in a space of predictable processes such to guarantee its integrability against $\bar{X}^h$ and $\bar{Y}^{h}$. It makes no harm in practice to also assume that $Y^{h,c}$ is a finite variation process 
for each $h$ and $c$, so that the intermediate and more intuitive representation in \eqref{eq:conversion_flows} is conveniently well defined pathwise in Riemann-Stieltjes sense, because virtually any model of real securities describes cash flows with continuous yields $y^{h,c}_t\diff{}t$ and/or discrete dividends $\Delta Y^{h,c}_t\delta(t-T_i^{h,c})$.

Furthermore and even more to the point, the numerical optimization entails a discretization of the set of decision times, so that for all integrals where the integrator is a truly random process, the integrand is in fact piecewise constant: hence, to give a sense to all formulas, one only needs the trivial definition of the stochastic integral for a simple process.
\end{remark}

\subsubsection{Constrained and free variables}\label{sec:constraints}

The set of admissible allocations $\vc{N}$ is subject to a set of financial constraints which we list here below.
\begin{enumerate}
\item No CVA transfer: the CVA item remains in the book for the whole optimization horizon.\footnote{Another implicit assumption is that the netting set for which CVA is calculated remains the same. When in reality some deal is added or unwound, one can always compute the correction by difference and charge it to the client.} This translates into the constraint $N^{\CVA}_t \equiv 1$.
\item Collateral balances: a (possibly empty) subset of the book items $\mathscr{A}\subseteq\mathscr{H}$ is comprised of collateral accounts. For $a\in\mathscr{A}$, let us denote by $\mathscr{H}_a\subset\mathscr{H}$ the set of instruments which are collateralized by account $a$. Then the amount $N_t^a$ is not a free variable: it is a function $ N^a(\vc{N}^{\mathscr{H}_a})$ of 
\[ 
    \vc{N}^{\mathscr{H}_a}:= (N^h)_{h\in\mathscr{H}_a}
\]
and of the path of market drivers via the Credit Support Annex (CSA) rules\footnote{It might even depend on the activity of other traders on the same account.
We argue in section~\ref{sec:collateral_balance} that a model for the collateral avoiding such non-modellable complication should be accurate enough for our purposes.}. Again we stay general in this section, leaving the choice of a model for such a map in the numerical application.
\item Self financing: no cash is injected or withdrawn, so that the value of the book $V_t$ equals the cumulated gains $G_t$. Enforcing such constraint is very easy if one has free access to a bank account $\bar{h}$ in currency $\bar{c}$, which has the property $\bar{X}^{\bar{h}}_t \equiv 1$ for all $t$: indeed, in such case it is sufficient to force
\begin{equation*}
	N^{\bar{h}}_t = V_{t_0}+\sum_{h\in\mathscr{H}}\int_{t_0}^t N_s^h \diff{}\bar{X}_s^h + \sum_{h\in\mathscr{H}} \int_{t_0}^t N_s^h \diff{}\bar{Y}_s^h - T_t(\vc{N}) - \sum_{h\in\mathscr{H}\setminus\{\bar{h}\}} N^{h}_t \bar{X}^h_t
\end{equation*}
where we recall that $t_0$ denotes the initial time.
\end{enumerate}
We conclude that the free variables for optimization are
\[
	\vc{N}^{\mathscr{F}}_t := (N^h_t)_{h\in\mathscr{F}}, \qquad \mathscr{F} :=\mathscr{H}\setminus(\mathscr{A}\cup\{\CVA,\bar{h}\}).
\]

\subsection{Financial instruments: price and dividend processes}\label{sec:processes}
In this section we describe in some detail the relevant financial instruments and how they fit in the mathematical framework described in section~\ref{sec:setting}. In particular, we describe the book value processes $X^h$, the dividend processes $Y^{h,c}$, the collateral balance processes $N^a = N^a(\vc{N}^{\mathscr{H}_a})$, and the cost process $T = T(\vc{N})$. 

\subsubsection{Credit Valuation Adjustment}\label{sec:CVA}

We identify by $h=\CVA$ the book item representing the Credit Valuation Adjustment to be hedged.

The price process $X^{\CVA}_t$ is the risk-neutral $t$-conditional expectation $\bbE^{\bbQ_t}_t$ of the loss $\LGD_\tau < 0$ recorded at the counterparty default time $\tau$ \citep{cesari2009modelling}:
\begin{equation*}
X^{\CVA}_t = I_{\tau > t}\bbE^{\bbQ_t}_t[D(t,\tau)\LGD_{\tau}]
\end{equation*}
computed with some pricing model $\bbQ_t$ and discount factor $D(t,\tau)$; note that it is by convention a non positive quantity. We allow the model $\bbQ_t$ to depend on time inconsistently with the real-world dynamics of $\LGD_{\tau}$, e.g.~because of periodic model recalibration, since pricing rules for a complex object like CVA typically gives up maximum realism due to model risk and computational feasibility concerns.

The dividend process $Y^{\CVA}$ is a jump process with a deterministic positive jump at $t=0$ equal to the counterparty risk premium paid by the counterparty to the bank (if any), and a random negative jump at default time $\tau$ equal to the loss given default $\LGD_{\tau}$. 

Details on $X^{\CVA}$ and $Y^{\CVA}$ for a concrete example portfolio can be found in the later section~\ref{sec:cva_fxfwd}.

\begin{remark}[Involved simplifications]
In the above modelling, we ignore the random delay between default time and default settlement; to put it differently, $\LGD_{\tau}$ should be more accurately interpreted as a valuation at time $\tau$ of the value of the investor's claim to the defaulted counterparty, at which time the position changes nature and its management gets outside the scope of this paper.
\end{remark}

\subsubsection{Credit Default Swaps}\label{sec:CDS}

We suppose that $\mathscr{H}$ can include Credit Default Swaps (CDS) on the CVA counterparty. By convention, we decide that $N_t^h > 0$ for such instruments indicates that the agent has a long position with respect to the risk (it sold protection).

The price process $X_t^h$ is the (dirty) price or upfront, which in real markets is not directly quoted, but is anyway computable from an observable quotation (the spread) through a conventional evaluation function \citep{isda2021cds}.

Denoting by $\tau$ the default time of the counterparty, the dividend process is a step function with positive jumps equal to the quarterly coupons $C$ on a schedule $T^{\CDS}_i$ contingent to counterparty survival, and a negative jump equal to the recovery flow $\Rec$ at default $\tau$:\footnote{We neglect the random delay from default time to default settlement.}
\begin{equation*}
Y^{h}_t = \sum_{T^{h}_i \leq t} C I_{\tau > T^{h}_i} - \ I_{\tau \leq t}(1-\Rec).
\end{equation*}

\begin{remark}[Involved simplifications]
In the above modelling, we suppose that the recovery flow is a constant equal to that adopted by market pricing conventions; more complex choices could be considered in future work, but we warn that they would be subject to relevant model risk, and may have little effect on the optimized trading strategy if deviations from the mean are mainly idiosyncratic and hence we cannot expect to offset them by dynamic hedging.
\end{remark}

\subsubsection{Cash accounts}\label{sec:cash}

We call ``cash account'' any element $b\in\mathscr{H}$ such that $X_t^b\equiv 1$ yielding a continuous dividend flow $\diff{}Y_t^b=r^b_t\diff{}t$.

We suppose that the basket $\mathscr{H}$ includes the following cash accounts:
\begin{itemize}
\item A funding account $f(c)\in\mathscr{H}$ for each currency $c\in\mathscr{C}$. Recall that the home funding $f(\bar{c})$ was already introduced in section~\ref{sec:constraints} with the notation $\bar{h}$;
\item The set $\mathscr{A} \subset \mathscr{H}$ of collateral accounts.
\end{itemize}

\begin{remark}[Involved simplifications]
In describing the above book items as idealized bank accounts, we are making the following minor simplifications:
\begin{enumerate}
\item We are postulating access to short-term funding in all currencies. In a typical bank, the treasury offers such facilities to traders, shadowing them from the optimization of the underlying external funding sources, which may involve FX instruments to mimic funding in foreign currencies.
\item Real funding accounts would pay interest no faster than daily. For simplicity, we suppose the payment frequency is high enough to be well approximated by a continuous stream.
\item Real collateral accounts typically compute earned interests on a daily frequency, but pay them at a lower frequency (typically monthly). Such delay may or may not be contractually compensated by a capitalization factor. Again we suppose for simplicity that such short delays are negligible.
\end{enumerate}
\end{remark}

\subsubsection{Collateral balance}\label{sec:collateral_balance}

For each collateral account $a$, we suppose that the collateral balance is updated instantaneously, so to equal the fair value of the underlying exposure in the collateral currency $c_a$ at each time $t$. In formulas:
\begin{equation*}
N^a_t(\vc{N}^{\mathscr{H}_a}) = -\sum_{h\in\mathscr{H}_a} N_t^h X_t^{h,c_a}.
\end{equation*}

\begin{remark}[Involved simplifications]
In the above modelling, we are making the following simplifications:
\begin{enumerate}
\item In practice, the collateral balance is updated in discrete time. However, the frequency of updates is high (daily), so that it should be well approximated by instantaneous resets.
\item Real CSA agreements include rules such as minimum transfer amounts, which make the exact amount of the collateral balance a path-dependent object. We deem them negligible details.
\item Real CSA agreements often include over-collateralization terms, e.g.~in the form of initial margins \citep{andersen2018margin}, which make the collateral balance a non-linear object. Such effects would be bank-wide and would therefore destroy the possibility to optimize the hedging netting set by netting set, which is highly desirable. This motivates our choice to exclude such terms from our simulator for the moment being. Including a linear proxy might be a possible future improvement.
\item Some CSA agreements allow to post multiple currencies or even non-cash instruments as collateral. Collateral optimization is a separate complex task, so that it is unreasonable to expect our agent to take it into consideration. Therefore, we suppose that the collateral rate $r^a$ already includes an estimate of the best yield one can get from this optionality.
\end{enumerate}
\end{remark}

\subsubsection{Rebalancing costs}\label{sec:costs}

To model rebalancing costs, we suppose that the following market operations are used to rebalance the portfolio:
\begin{itemize}
\item Non-cash assets ($h \neq f(c)$ for all $c\in \mathscr{C}$) can be exchanged for cash in the asset currency $c_h$. For those which carry negligible transaction costs we put $T_t^h \equiv 0$; for the other ones, we assume that the market quotes a bid price $X_t^{h,-}$ and an ask price $X_t^{h,+}$ valid for small transactions, and define the mid price and semi-spread as
\begin{equation}\label{eq:bidask}
    X^h_t = \frac{X^{h,+}_t+X^{h,-}_t}{2}, \qquad
    \gamma^h_t = \frac{X^{h,+}_t-X^{h,-}_t}{2}.
\end{equation}
With these notations, we adopt the following model for the costs due to transactions in $h$, denominated in currency $c_h$:
\begin{equation}\label{eq:costs}
	T_t^h = \int_{t_0}^t \gamma^h_s |\diff{}N^h_s| + \sum_{s\leq t} \alpha^h \gamma^h_s |\Delta N^h_s|^2,
\end{equation}
where non infinitesimal transactions get a worse price due to an impact term whose proportionality to trade size is measured by the constant~$\alpha^h$.\footnote{Note that $|dN^h_s|$ denotes total variation: in particular the $h$ component of the trading strategy $\vc{N}$ must be a finite variation process for all assets $h$ with non negligible costs.} Its conversion to home currency equals
\begin{equation*}
\bar{T}_t^h := \int_{t_0}^t \phi_s^{c_h\bar{c}} \diff{}T_s^h = \int_{t_0}^t \phi_s^{c_h\bar{c}} \gamma^h_s |\diff{}N^h_s| + \sum_{s\leq t} \phi_s^{c_h\bar{c}} \alpha^h \gamma^h_s |\Delta N^h_s|^2.
\end{equation*}
\item Foreign funding accounts $f(c)$ for $c \neq \bar{c}$ can be converted to domestic cash $\bar{h}$ with spot operations on the foreign exchange market. The cumulated notional $F^c_t$ in $c$ currency of such operations up to time $t$ excluded must be computed net of all other $c$ denominated cash flows, so it satisfies by definition
\[
	N_{t-}^{f(c)} = F^c_t - \sum_{\substack{h\in \mathscr{H}\setminus\{f(c)\}\\c_h = c}} T_t^h + \sum_{h\in \mathscr{H}} \int_{t_0}^{t-} N^h_s \diff{}Y_s^{h,c},
\]
which implicitly defines $F^c_t$. Now in analogy with \eqref{eq:bidask}-\eqref{eq:costs}, we introduce bid and ask values for the exchange rate, respectively $\phi^{c\bar{c},-}$ and $\phi^{c\bar{c},+}$, and put
\begin{gather}
    \phi^{c\bar{c}}_t = \frac{\phi^{c\bar{c},+}_t+\phi^{c\bar{c},-}_t}{2}, \qquad
    \gamma^c_t = \frac{\phi^{c\bar{c},+}_t-\phi^{c\bar{c},-}_t}{2}\\
	\bar{T}_t^{f(c)} = \int_{t_0}^t \gamma^c_s |\diff{}F^c_s| + \sum_{s\leq t} \alpha^c \gamma^c_s |\Delta F^c_s|^2.
\end{gather}
Note that $\bar{T}_t^{f(c)}$ is already expressed in evaluation currency $\bar{c}$, since we supposed that all FX transactions are against it.
\end{itemize}

Finally, we can compute the total transaction costs as
\[
	T_t(\vc{N}) := \sum_{h \in \mathscr{H}} \bar{T}_t^h.
\]

\begin{remark}[Involved simplifications]
In the above modelling, we are supposing that foreign currencies are always exchanged against the reference currency and never between themselves. This is mainly for notational convenience, and does not affect our numerical study, which involves only one foreign currency.
\end{remark}

\subsection{CVA of an FX forward}\label{sec:cva_fxfwd}

In this section we specialize the CVA item of section \ref{sec:CVA} for the case of an FX forward. In particular, we introduce the payoff and market pricing of FX forwards, then we provide some details of the CVA computation.

An FX forward is a financial contract according to which two counterparties agree to exchange at a future date $T$ an amount $N^{c_1}$ in a currency $c_1$ against an amount $N^{c_2}$ in currency $c_2$. Both amounts are known at contract inception. We adopt the point of view of the counterparty receiving the $c_1$-denominated flow, so that the fair value of the contract at time $t<T$ is
\begin{equation*}
    E_t := \bbE_t^{\bbQ} \left[D^{\bar{c}}(t,T) \left(N^{c_1}\phi^{c_1\bar{c}}_T - N^{c_2}\phi^{c_2\bar{c}}_T\right)\right]
\end{equation*}
where $D^{\bar{c}}(t,T)$ denotes pathwise discounting from time $T$ to time $t<T$ using some reference interest rate in currency $\bar{c}$, and $\bbQ$ is a consensus pricing model. It is customary to express this price as
\begin{equation}\label{eq:exposure_2}
    E_t := \phi^{c_1\bar{c}}_t P_t^{c_1}(T) N^{c_1} - \phi^{c_2\bar{c}}_t P_t^{c_2}(T) N^{c_2}
\end{equation}
where the curves $P^{c_i}_t$ in currency $c_i$ at time $t$ satisfy by construction
\[
    P^{c_i}_t(T)\phi^{c_1\bar{c}}_t = \bbE^{\bbQ}\left[D^{\bar{c}}(t,T) \phi^{c_1\bar{c}}_T \right] \qquad \forall T \geq t.
\]

From now on, we consider the CVA of an uncollateralized portfolio consisting of a single FX forward where $c_2$ is the evaluation currency $\bar{c}$, as a prototypical example of a position whose counterparty risk depends on the dynamics of exchange rates. As is customary in practice, we suppose that the loss at default $\LGD_\tau$ is a constant fraction $(1-\Rec)$ of the positive exposure $E_{\tau}^+$ at default time $\tau$:
\[
    \LGD_\tau := -(1-\Rec) \max(E_{\tau},0)
\]
where $\Rec$ is the same conventional recovery fraction used in CDS pricing and introduced in section~\ref{sec:CDS}.

The book valuation of the CVA at time $t$ requires the specification of a joint $\bbQ_t$ model for the factors which determine the default time $\tau$ and the price \eqref{eq:exposure_2} at such time. A typical choice would be a Cox model for instantaneous default probabilities, so that 
\begin{equation}\label{eq:cva_cox}
    \CVA_t := -(1-\Rec) \bbE_{t}^{\bbQ_t} \left[\int_t^T D^{\bar{c}}(t,s) \max(E_{s},0) \lambda_s e^{-\int_t^s \lambda_u \,\diff{}u} \,\diff{}s  \right].
\end{equation}
for an instantaneous default intensity process $\lambda_s$.

Formula \eqref{eq:cva_cox} allows for dependence between default intensities and exposures; it is in fact not uncommon to assume them independent in pricing, so that the CVA item simplifies to
\begin{equation}\label{eq:cva_indep}
    \CVA_t := -(1-\Rec) \int_t^T \bbE_{t}^{\bbQ_t} \left[D^{\bar{c}}(t,s) \max(E_{s},0) \right] \bbQ_t(\tau \in \diff{}s).
\end{equation}
One advantage of such approach is that both the integrand and the unconditional risk neutral default probabilities $\bbQ_t(\tau \in ds)$ can be bootstrapped from reasonably liquid quotes. In fact, if one is willing to overlook the stochasticity of interest rates in CVA pricing, then
\begin{equation}\label{eq:cva_determ_rates}
\bbE_t^{\bbQ_t}\left[D^{\bar{c}}(t,s) \max(E_{s},0) \right] = P^{\bar{c}}_t(s) P_s^{c_1}(T) N^{c_1} \bbE_t^{\bbQ_t}\left[\max\left(\phi^{c_1\bar{c}}_s - \frac{P_s^{\bar{c}}(T) N^{\bar{c}}}{P_s^{c_1}(T) N^{c_1}}, 0\right)\right]
\end{equation}
can be interpreted as a deterministic multiple of the payoff of a vanilla call on the FX rate $\phi^{c_1\bar{c}}_s$ with suitable strike, and directly marked to the market of standard FX options at time $t$.

\begin{remark}[Coherence of simulator and pricer] While our modelling of the trader's portfolio in section~\ref{sec:processes} is as accurate as possible in terms of costs and as flexible as possible in terms of the probability law of the risk drivers, the pricing setup ignores costs altogether and introduces some perhaps crude simplifications one might adopt in the pricing model $\bbQ_t$. This paradoxically enhances the realism of our description. Indeed, nobody guarantees that the profit-and-loss of the trader will be measured with a risk-neutral model which is a perfect description of reality, and what the trader wants to model \emph{accurately} is exactly the distribution of this exogenously defined and maybe \emph{simplified} PnL. Of course they may later induce the firm to improve the pricing model; but we aim to investigate numerically whether the agent can learn a decent hedging strategy in terms e.g.~of hedging costs and correlations even when the book value of CVA ignores both aspects.
\end{remark}

\section{Reinforcement Learning and its application to hedging}\label{sec:RL}

Reinforcement Learning (RL) \citep{sutton1988learning} is a machine learning framework for sequential decision-making processes, and as such we deem it suitable for our CVA hedging problem. In the present section we briefly introduce classical RL concepts and some recent risk-averse variants thereof, to adapt them to a random finite horizon setting like ours. Finally, we explain how an hedging task can be formulated as an RL problem.

\subsection{Markov Decision Process}\label{sec:MDP}
The basic building block to apply RL algorithms to a problem is a description of the latter as a Markov Decision Process (MDP) \citep{puterman2014markov}.

\begin{definition}[Markov Decision Process]
A discrete-time MDP is defined as a 6-tuple $\mathcal{M} = \langle\Sspace,\Aspace, \mathcal{P}, \mathcal{R}, \gamma, \mu \rangle$, where: 
\begin{itemize}
\item $\Sspace$ is a non-empty measurable space called state space;
\item $\Aspace$ is a non-empty measurable space called action space;
\item $\mathcal{P}:\Sspace\times \Aspace \rightarrow P(\Sspace\times\mathbb{R})$ is the transition model that assigns to each state-action pair $(s,a)$ the probability measure $\mathcal{P}(\cdot|s,a)$ of the next state;
\item $\mathcal{R}:\Sspace\times \Aspace \times \Sspace\rightarrow P(\mathbb{R})$ is a bounded reward model, which assigns for every triple $(s,a,s')$ a probability measure $\mathcal{R}(\cdot|s,a,s')$; 
\item $\gamma\in[0,1]$ is the discount factor, used to weight future rewards;
\item $\mu$ is the initial state distribution, from which the starting state is sampled. 
\end{itemize}
\end{definition}

The interpretation of the above objects is the following.

\begin{description}

\item[State] The state $s \in \Sspace$ usually contains all the information the agent perceives from the environment. In a finite horizon setting like ours, it should also include a timestamp.

\item[Action] The action $a \in \Aspace$ represents how the agent interacts with the environment.

\item[Dynamics] $\mathcal{P}(\diff{}s'|s,a)$ is the probability of reaching state $s'$ given that we are in state $s$ and take action $a$. The environment dynamics fulfil the Markov Property, which means that state transitions depend only on the most recent state and action and not on previous history.

\item[Reward] $\mathcal{R}(\diff{}r|s,a,s')$ is the probability that the contribution to the agent's goal received when performing action $a\in \mathcal{A}$ in state $s \in \mathcal{S}$ and arriving in state $s' \in \mathcal{S}$ is equal to $r$.
We can also consider its expectation given the transition:
\[
	r(s,a,s') = \int_\mathbb{R}r\mathcal{R}(\diff{}r|s, a, s'),
\]
and also an expected reward independent of the next state $s'$, by computing an expectation over the next state:
\[
	r(s,a)=\int_\mathcal{S}r(s,a,s')\mathcal{P}(\diff{}s'|s,a).
\]

\end{description}

\subsection{Value functions and Bellman equations} \label{sec:value_fun}
We consider finite horizon problems in which future rewards are exponentially discounted with~$\gamma$. 
Let us define a trajectory as a sequence of states, actions, and rewards, up to a stopping time $\eps$:
\begin{equation}\label{eq:trajectory}
(s_{0}, a_{0}, r_{1}, s_{1}, a_{1}, r_{2}, ..., s_{\eps-1}, a_{\eps-1}, r_{\eps}).
\end{equation}

\begin{remark}[Termination time]
The episode termination time-step $\eps$ can be modelled without loss of generality as the first step at which the state $s_\eps$ would enter an absorbing termination region $\Tspace\subset \Sspace$, so that its law is included in the definition of $\mathcal{P}$. We suppose that $\eps>0$ $\mu$-almost surely; moreover, if $\gamma = 1$, we require that $\eps$ is almost surely finite for every choice of $\pi$.\footnote{For $\gamma<1$ one can be easily weaken the requirement thanks to the exponential decay of $\gamma^N$, but this is irrelevant for our purposes.} In our application, $\eps$ corresponds to the earlier between counterparty default and a maximum optimization horizon, which may coincide or not with portfolio maturity.
\end{remark}

We then define the discounted sum of the rewards of a trajectory:
\begin{equation*}
    \mathcal{G} = \sum_{i=1}^{\eps} \gamma^{i-1} r_{i}.
\end{equation*}
Each trajectory is generated by following a \textit{policy} $\pi: \Sspace\rightarrow P(\Aspace)$, mapping each state $s$ to a probability distribution $\pi(\cdot|s)$ on actions. This formalizes mathematically a (possibly randomized) strategy the agent plays to select the action at each time-step. It determines the distribution of trajectories \eqref{eq:trajectory}: therefore, expected values over such distribution is denoted by the symbol~$\EVp$.

Given that both the policy and the transition probability may be stochastic, we are interested in the expected value of the return given all the possible trajectories, also known as the value function.
\begin{definition}[State value function or V-function]
Let $\mathcal{M}$ be an MDP and $\pi$ a policy. For every state $s \in \Sspace$, the state value function $V_\pi:\Sspace \rightarrow \mathbb{R}$ is defined as the expected return starting from state $s$ and following policy $\pi$:
\begin{equation*}
    V_\pi(s) \coloneqq \EVp \left[\mathcal{G}\Big|s_0 = s\right],
\end{equation*}
\end{definition}

\begin{remark}[Bellman equation for $V$]  $V_\pi(s) = 0$ trivially for $s\in\Tspace$, while for the non-trivial case $s\notin\Tspace$, the state value function can be recursively determined by \citep{bellman1966dynamic}:
\begin{equation*}
	V_\pi(s) = \EV_{a \sim\pi(\cdot|s)}\left[r(s,a) +\gamma \EV_{s'\sim \mathcal{P}(\cdot|s,a)}\big[V_\pi(s')\big]\right].
\end{equation*}
\end{remark}

Similarly, if we consider starting from a specific state and taking a specific action, we can define the state-action value function.
\begin{definition}[State-action value function or Q-function]
Let $\mathcal{M}$ be an MDP and $\pi$ a policy. For every state-action pair $(s,a)\in \Sspace \times \Aspace$, the state-action value function $Q^\pi:\Sspace\times \Aspace \rightarrow \mathbb{R}$ is defined as the expected return starting from state $s$, playing action $a$ and following policy $\pi$:
\begin{equation}\label{eq:Q_fun}
	Q_\pi(s,a) \coloneqq \EVp\left[\mathcal{G}|s_0 = s, a_0 = a\right]
\end{equation}
\end{definition}

\begin{remark}[Bellman equation for $Q$]
Again for the non-trivial case $s\notin\Tspace$, the state-action value function can be recursively defined by the following \citep{bellman1966dynamic}:
\begin{equation}\label{eq:bellman}
	Q_\pi(s,a) = r(s,a) + \gamma \EV_{\substack{s'\sim \mathcal{P}(\cdot|s,a)\\a'\sim\pi(\cdot|s')}}\big[Q_\pi(s',a')\big].    
\end{equation}
\end{remark}

Finally, the objective in \textit{risk neutral} RL is the maximization of the value function, given an initial state distribution.

\begin{definition}[Unnormalized expected return]
Let $\mathcal{M}$ be an MDP and $\pi$ a policy. Given the initial state distribution $\mu$, the unnormalized expected return $\hat{J}_\pi$ is defined as the expectation of the return starting from a $\mu$-distributed initial state $s_0$ and following policy $\pi$:
\begin{align*}
	\hat{J}_\pi \coloneqq \EVp_{s_0\sim \mu}\left[\mathcal{G}\right].
\end{align*}
\end{definition}

We also introduce its normalized version.

\begin{definition}[Normalized expected return]
Let $\mathcal{M}$ be an MDP and $\pi$ a policy. Given the initial state distribution $\mu$, the normalized expected return $J_\pi$ is defined as the expectation of the discount factor-weighted average of the rewards starting from a $\mu$-distributed initial state $s_0$ and following policy $\pi$:
\begin{align*}
    J_\pi & \coloneqq \EVp_{s_0\sim \mu}\left[\Gamma^{-1}\mathcal{G}\right], \qquad \Gamma = \sum_{i=1}^{\eps} \gamma^{i-1}.
\end{align*}
\end{definition}

\begin{remark}[Differences from infinite horizon]
The normalization factor $\Gamma$ is chosen so that the weighting $\left(\gamma^{i-1}\Gamma^{-1}\right)_{i=1,\dots,\eps}$ is a probability measure on the set of time steps $\{1,\dots,\eps\}$.
The original definition was written for an infinite horizon problem, hence it required $\gamma<1$ and used $\Gamma = (1-\gamma)^{-1}$; the above is adapted to a finite random horizon, allowing for $\gamma\leq 1$ and giving
\[
	\Gamma =\frac{1-\gamma^\eps}{1-\gamma} \text{ if }\gamma < 1, \qquad \Gamma = \eps \text{ if }\gamma = 1.
\]
\end{remark}

\subsection{Volatility aversion}\label{sec:TRVO}
A number of modified risk-aware objectives have been studied, for example introducing a trade-off with the minimization of variance of the returns, in a mean-variance \citep{tamar2013variance, prashanth_actor-critic_2014} or Sharpe ratio \citep{moody2001learning} fashion. Others have studied the minimization of CVaR or more generally of a coherent risk measure \citep{tamar_sequential_2017}.

Nevertheless, all these approaches consider only the minimization of the long-term risk, while in financial trading interim results are also fundamental, and keeping a low-varying intermediate PnL becomes crucial. Moreover, the analytical intractability of all these formulation does not allow the related algorithms to perform (in terms of learning improvements) as the state-of-the-art algorithms in the standard RL framework, such as Trust Region Policy Optimization (TRPO)~\citep{schulman2015trust}. For these reasons, \citet{Bisi2020trvo} introduced a new measure of risk, which takes into account the variance of the reward at each time-step with respect to state visitation probabilities:
\begin{definition}[Unnormalized reward volatility]
The unnormalized reward volatility a.k.a  unnormalized reward variance is expressed as:
\begin{equation*}
	\hat{\nu}^2_\pi = \EVp_{{s_0 \sim \mu }}\left[\sum_{i=1}^{\eps} \gamma^{i-1} \left(r_i-J_\pi\right)^2\right]
\end{equation*}
\end{definition}

\begin{remark}[Similarity to PnL volatility]
The argument of the expected value is akin to a the variance of the immediate rewards across \emph{time}, but for the fact that the residuals are expressed with respect to a value $J_\theta$ which is an average across both population and time. As such, it is comparable to the concept of profit-and-loss volatility used in financial practice.
\end{remark}

\begin{remark}[Differences from infinite horizon]
A factor $\Gamma^{-1}$ inside the expected value would make $\hat{\nu}_\pi^2$ a true variance across population and time, giving the (normalized) return volatility $\nu_\pi^2$ defined in the original infinite horizon formulation. Here we remove that factor because we believe that with stochastic episode lengths, an underweighting of returns belonging to longer episodes would be financially inappropriate. We also suspect that if we kept that factor, we could not have a result like theorem~\ref{thm:gradient}, which is crucial for the learning algorithm. Note that when $\Gamma$ is deterministic as in the original paper, putting or removing a multiplier in the definition of $\nu_\pi^2$ is just a matter of notation, since it can be absorbed in the risk aversion coefficient $\beta$ defined here below.
\end{remark}

In most trading and even hedging applications, achieving a profit is at least as relevant as being risk-averse thus, we decide to consider an objective that handles the risk-return trade-off through a risk aversion coefficient, the parameter $\beta$. The objective related to the policy $\pi$ can be defined as:
\begin{equation*}
	\hat{\eta}_\pi \coloneqq \hat{J}_\pi - \beta \hat{\nu}^2_\pi,
\end{equation*}
called \textit{unnormalized mean-volatility} hereafter, where $\beta\geq 0$ allows to trade-off expected return maximization with risk minimization. Once more, the original normalized goal $\eta_\pi \coloneqq J_\pi - \beta \nu^2_\pi$ is just a multiple of ours in deterministic horizon settings. On the other hand, the more standard mean-variance objective is $\hat{J}_{\pi} - \beta\sigma_{\pi}^2$, where $\sigma_{\pi}$ is the \emph{return} variance as defined in~\citet{tamar2013variance}:
\begin{definition}[Return variance]
The return variance is expressed as:
\begin{equation*}
	\sigma_\pi^2 \coloneqq \EVp_{s_0 \sim \mu}\left[\left( \mathcal{G}-\hat{J}_\pi\right)^2\right].
\end{equation*}
\end{definition}

Here below we generalize an important result on the relationship between the two variance measures:
\begin{lemma}[Variance inequality]\label{thm:variance_ineq}
The following inequality holds:
\begin{equation*}
 \sigma_\pi^2 \leq \esssup(\Gamma)\hat{\nu}^2_\pi.
\end{equation*}
\end{lemma}
\begin{proof}
By factoring $\Gamma$, we can interpret the inner summation on $i$ in the below equation as an expected value, and apply to it the Jensen inequality:
\begin{align*}
\sigma_\pi^2 &= \EVp_{s_0 \sim \mu}\left[\Gamma^2\left(\sum_{i=1}^{\eps}\Gamma^{-1}\gamma^{i-1}(r_i - J_\pi)\right)^2 \right]\\
                   & \leq \EVp_{s_0 \sim \mu}\left[\Gamma^2\sum_{i=1}^{\eps}\Gamma^{-1}\gamma^{i-1}\left(r_i - J_\pi\right)^2 \right] \leq \esssup(\Gamma) \hat{\nu}^2_\pi.
\end{align*}
\end{proof}

\begin{remark}[On the essential supremum]
The $\esssup$ in the Lemma is for sure finite if the maximum length of an episode (under policy $\pi$) is finite, as in our application. Otherwise, it is finite only if $\gamma<1$; in which case it is smaller than $(1-\gamma)^{-1}$.
\end{remark}

The $\Gamma$-related factor just comes from the fact that the return variance is not normalized, unlike the reward volatility. Apart from that, the key difference between the two measures is the different role of inter-temporal correlations between the rewards.
However, Lemma~\ref{thm:variance_ineq} shows that the minimization of the reward volatility yields a low return variance. The opposite is clearly not true: as counterexample, it is possible to consider a stock price with the same value at the beginning and at the end of the investment period, but making complex movements in-between.

It turns out that the mean-volatility objective, besides being closer to what practitioners actually monitor, is also analytically more tractable than a risk measure on cumulated rewards. Indeed, one can derive linear Bellman equations, in a similar form as in equation~\eqref{eq:bellman}, and a policy gradient theorem, analogous to the standard RL result. 

Indeed, we can introduce a volatility equivalent of the action-value function $Q_\pi$ \eqref{eq:Q_fun}, called \textit{action-volatility} function, which is the volatility observed by starting from state $s$, taking action $a$, and following policy $\pi$ thereafter:
\begin{equation*}
X_\pi(s,a) \coloneqq \EVp\left[\sum_{i=1}^{\eps} \gamma^{i-1} (r_i - J_\pi)^2|s=s_0,a=a_0\right].
\end{equation*}
Like the $Q_\pi$ function, this can be written recursively by means of a Bellman equation:
\begin{equation}\label{eq:bellmanX}
	X_\pi(s,a) = \int_\mathcal{S}\mathcal{P}(\diff{}s'|s,a)\int_{\mathbb{R}}\mathcal{R}(\diff{}r|s,a,s')\left[r-J_\pi\right]^2 + \gamma \EV_{\substack{s'\sim P(\cdot|s,a)\\a'\sim\pi(\cdot|s')}}\big[X_\pi(s',a')\big],
\end{equation}
where often the first integral is well approximated by $\left[r(s,a)-J_\pi\right]^2$.

Now we can finally derive the policy gradient:
\begin{theorem}\label{thm:gradient}
Consider policies which are absolutely continuous with respect to a reference measure $\bar{\pi}$, such that the density $\pi_\theta(\cdot|s)$ depends differentiably on real-valued parameters~$\theta$. 
Then (under technical hypotheses: see the next Remark) we can express the gradient of the squared reward volatility as a pathwise expected value:
\begin{equation}\label{eq:gradient}
	\nabla_\theta\hat{\nu}_\pi^2 = \EVp \left[ \sum_{i=0}^{\eps-1} \gamma^{i}X_\pi(s_i,a_i)\nabla_\theta\log \pi_{\theta}(a_i|s_i) \right].
\end{equation}
\end{theorem}
\begin{proof}
We prove by induction on $N$ that
\begin{multline}\label{eq:gradient_induction}
	\nabla_\theta \hat{\nu}_\pi^2 = \EVp \left[ \sum_{i=0}^{N-1} I_{\eps > i} \gamma^{i} X_{\pi}(s_i,a_i)\nabla_\theta \log\pi_\theta(a_i|s_i) \right.\\
	\left. - 2\left(\nabla_\theta J_{\pi}\right)  \sum_{i=1}^{N} I_{\eps \geq i} \gamma^{i-1}(r_i - J_\pi) + I_{\eps \geq N} \gamma^N \nabla_\theta \EV_{\substack{s'\sim P(\cdot|s_{N-1},a_{N-1})\\a'\sim\pi(\cdot|s')}} \left[ X_\pi(s',a') \right] \right]. 
\end{multline}
Letting $N\to\infty$ would then give the conclusion because by hypothesis $\eps$ is almost surely finite, hence:
\begin{enumerate}
\item The first addend under $\EVp$ is eventually equal to that in \eqref{eq:gradient};
\item The second addend is eventually equal to $-2\nabla J_{\pi} \sum_{i=1}^{\eps} \gamma^{i-1}(r_i - J_\pi)$ where the sum equals $0$ by definition of $J_\pi$;
\item The third addend is eventually null because of the indicator $I_{\eps \geq N}$.
\end{enumerate}

In the base case $N=0$, and slightly abusing the notation $\mathcal{P}(\cdot|s_{-1},a_{-1})$ to mean the initial distribution $\mu$, the two sums in \eqref{eq:gradient_induction} are empty and we must check
\begin{equation*}
	\nabla_\theta \hat{\nu}_\pi^2 = \EVp \left[\nabla_\theta \EV_{\substack{s'\sim \mu\\a'\sim\pi(\cdot|s')}} \left[ X_\pi(s',a') \right] \right],
\end{equation*}
where the outer expected value is useless because the integrand is deterministic: we must only prove
\begin{equation*}
	\nabla_\theta \hat{\nu}_\pi^2 = \nabla_\theta \EV_{\substack{s'\sim \mu\\a'\sim\pi(\cdot|s')}} \left[ X_\pi(s',a') \right],
\end{equation*}
which is true even without the $\nabla_\theta$ operator, by the definitions of $\nu_\pi$ and $X_\pi$.

In the induction step we assume \eqref{eq:gradient_induction} for some $N$ and prove it for $N+1$. To this purpose, we expand the third addend, and differentiate under expectation and integral signs:
\begin{multline*}
\nabla_\theta \EV_{\substack{s'\sim P(\cdot|s_{N-1},a_{N-1})\\a'\sim\pi(\cdot|s')}} \left[ X_\pi(s',a') \right] = 
\nabla_\theta \EV_{s'\sim P(\cdot|s_{N-1},a_{N-1})} \left[\int_{\Aspace} X_\pi(s',a') \pi_\theta(a'|s')\bar{\pi}(\diff{}a') \right]\\
 =\EV_{s'\sim P(\cdot|s_{N-1},a_{N-1})} \left[\int_{\Aspace} \left( X_\pi(s',a') \frac{\nabla_\theta \pi_\theta(a'|s')}{\pi_\theta(a'|s')}+ \nabla_\theta X_\pi(s',a')\right)\pi_\theta(a'|s')\bar{\pi}(\diff{}a') \right] \\
= \EV_{\substack{s'\sim P(\cdot|s_{N-1},a_{N-1})\\a'\sim\pi(\cdot|s')}}  \left[X_\pi(s',a') \nabla_\theta \log \pi_\theta(a'|s')+ \nabla_\theta X_\pi(s',a') \right].
\end{multline*}
We recall that the gradient we just computed appeared in \eqref{eq:gradient_induction} pre-multiplied by $\gamma^N I_{\eps \geq N}$ in an expected value $\EVp$, and $(s'\sim P(\cdot|s_{N-1},a_{N-1}), a'\sim\pi(\cdot|s'))$ defines exactly the conditional law of $(s_N,a_N)$ given $(s_{N-1},a_{N-1})$: by the tower law we deduce that
\begin{multline*}
\EVp \left[ \gamma^N I_{\eps \geq N} \nabla_\theta \EV_{\substack{s'\sim P(\cdot|s_{N-1},a_{N-1})\\a'\sim\pi(\cdot|s')}} \left[ X_\pi(s',a') \right] \right]\\
= \EVp \left[\gamma^N I_{\eps \geq N} X_\pi(s_N,a_N)  \nabla_\theta \log \pi_\theta(a_N|s_N) \right] + \EVp \left[ \gamma^N I_{\eps \geq N} \nabla_\theta X_\pi(s_N,a_N) \right].
\end{multline*}
Since $X_\pi(s_N,a_N)=0$ on $\{\eps=N\}$, we can substitute the indicators $I_{\eps \geq N}$ with  $I_{\eps > N}$ in the first addend and with the equivalent $I_{\eps \geq N+1}$ in the second. The first expectation above then becomes the $i=N$ term in the first summation of \eqref{eq:gradient_induction}, and we only have to analyse the second expectation.

Now we invoke the Bellman-like recursion \eqref{eq:bellmanX} and get
\begin{multline*}
 \EVp \left[ \gamma^N I_{\eps \geq N+1} \nabla_\theta X_\pi(s_N,a_N) \right] =\\
 \EVp \left[ -2\gamma^N I_{\eps \geq N+1}\int_\mathcal{S}\mathcal{P}(\diff{}s'|s_N,a_N)\int_{\mathbb{R}}\mathcal{R}(\diff{}r|s_N,a_N,s')\left[r-J_\pi\right]\nabla_\theta J_\pi\right] \\
+ \EVp \left[\gamma^N I_{\eps \geq N+1} \nabla_\theta \gamma \EV_{\substack{s'\sim P(\cdot|s_N,a_N)\\a'\sim\pi(\cdot|s')}}\big[X_\pi(s',a')\big] \right]
\end{multline*}
In the double integral on the right hand side, $r$ is sampled exactly according to the law of $r_{N+1}$ given $(s_N,a_N)$, so we are done proving \eqref{eq:gradient_induction} with $N$ replaced by $N+1$.
\end{proof}

\begin{remark}[Technical hypotheses]
The proof of the theorem relies on a passage to the limit $N\to\infty$ within the expectation \eqref{eq:gradient_induction}, and on a couple of exchanges of integrals with the $\nabla_\theta$ operator. This is why the theorem is stated under unspecified technical hypotheses: many elementary results would justify the mentioned manipulations, so we believe that picking one specific set of assumptions for the sake of formalism would add little to our treatment.
\end{remark}

With the above result at our disposal, generalization of the Trust Region Volatility Optimization (TRVO) algorithm of \citet{Bisi2020trvo} is straightforward, and hence omitted: we are eventually in the position of using an efficient risk-averse optimization tool to solve our specific control problem.

\subsection{Hedging as a Reinforcement Learning problem}\label{sec:RL_for_hedging}

By the definition of the gain process in \eqref{eq:gain}, the signed increment $G_t-G_s$ represents the performance over the time period $[s,t)$ of the hedging strategy, usually referred to as profit-and-loss (PnL). Traders aim at maximizing such increment, but also at controlling its variability, not only in distributional sense as the possibility of a large negative PnL over the full trading period $[t_0,\bar{t})$, but also in the time direction: e.g., they cannot accept recording a large loss $G_{\bar{t}/2} \ll 0$, even if it leads eventually to a positive gain $G_{\bar{t}}>0$ with high probability. All of this suggests that a good description of the trader's objective should consider a \emph{set} of increments over a time grid $t_0<t_1<\dots<t_N=\bar{t}$:
\begin{equation}\label{eq:reward_as_PnL}
	R_{t_{i+1}} = G_{t_{i+1}}-G_{t_i},
\end{equation}
to be maximized in volatility-averse sense. With this observation, we are ready to describe the hedging problem as an MDP with the following interpretations:

\begin{description}

\item[State] A set of variables $\vc{S}_{t_i}$ sufficient to determine all the financial processes defined in section~\ref{sec:environment} at a time $t_i$, and the law of their future evolution, plus a description of the portfolio position at $t_i$, and the timestamp $t_i$ itself. In formulas, we put $s_i \coloneqq \vc{S}_{t_i}$.

\item[Action] A set of variables $\vc{A}_{t_i}$ sufficient to represent the new composition of the portfolio in the time frame $(t_i,t_{i+1}]$, where we suppose that $N^{h}$ for $h\in\mathscr{F}$ changes only at the finite set of times $\{t_0,\dots,t_N\}$. In formulas, we put $a_i \coloneqq \vc{A}_{t_i}$.

\item[Dynamics] It encompasses the law of the evolution of the market on $[t_i,t_{i+1})$, and the mechanics of the book values as described in section~\ref{sec:environment}.

\item[Reward] It is defined as the PnL over a single period. In formulas, we put $r_{i+1} \coloneqq R_{t_{i+1}}$.

\item[Termination] It is the earliest between the trading horizon $\bar{t}$ and default time $\tau$ (discretized on the time grid). In formulas, we put $\eps \coloneqq \min\{i \colon t_i \geq \tau \wedge t_N\}$.

\end{description}

\begin{remark}[Role of the pricing model]
The return of an episode
\begin{equation}\label{eq:return}
    \mathcal{G} = \sum_{i=1}^{t_N} \gamma^{i} R_{t_i}.
\end{equation}
collapses telescopically to $G_{\bar{t}}$ if $\gamma=1$: in particular, maximizing \eqref{eq:return} means maximizing the eventual profit without risk aversion; moreover, if the optimization horizon $\bar{t}$ is the maturity of the portfolio, then there is no dependence of the objective on the pricing model. The latter plays a role only when either $\gamma < 1$ (encoding a preference in PnL timing, be it real or a bias introduced for better algorithm convergence), or when the time distribution of the PnL is part of the risk aversion as in TRVO.
\end{remark}

\section{Implementation}\label{sec:impl}

This section describes a set of concrete modelling choices we made to test numerically the above approach. We aimed at the simplest possible setup retaining the essential financial aspects of the hedging problem.

\subsection{Financial setting}\label{sec:financial_impl}
In this subsection we specify the financial environment with the general notation of section~\ref{sec:environment}.

\subsubsection{Assets and currencies}\label{sec:assets_impl}
The hedged CVA is due to a single EURUSD FX forward.
The set of currencies which appear in the book is
\[
    \mathscr{C} = \left\{ \EUR, \USD \right\}
\]
and EUR is considered to be the main evaluation currency, i.e.
\[
	\bar{c} = \EUR.
\]

We suppose that $\mathscr{H}$ includes one or more Credit Default Swaps on the CVA counterparty, which we identify with the notation $\CDS_m$ for $m\in\mathscr{M}$, where in practice the index set $\mathscr{M}$ consists of maturity labels, e.g.~$m=\mathrm{5Y}$ for a 5-years long CDS.\footnote{For long trading horizons, one may want to consider synthetic rolling instrument, to use always the most liquid on-the-run maturities for hedging. In such case, the dividend process should be carefully defined to include the roll costs.} Since the only collateralized assets are the CDS, the set of assets is therefore
\[
    \mathscr{H} = \left\{ \CVA, f(\EUR), f(\USD) \right\} \cup \left\{\CDS_m, a(\CDS_m)\right\}_{m\in\mathscr{M}},
\]
where $a(\CDS_m)$ denotes the collateral account of $\CDS_m$, which we assume to be EUR denominated:
\[
    Y_t^{\CDS_m, c} = Y_t^{a(\CDS_m), c} = 0, \qquad \text{if } c \neq \EUR.
\]
Moreover, since the free variables are those in the set $\mathscr{H}\setminus(\mathscr{A}\cup\{\CVA,\bar{h}\})$, it comes that the free set is actually
\[
	\mathscr{F} = \left\{ f(\USD) \right\} \cup \left\{\CDS_m\right\}_{m\in\mathscr{M}}.
\]
which we call informally ``hedging assets''.

We also assume that interest rates are known time-independent constants for all cash accounts.

\subsubsection{Data generation}\label{sec:simulator}
The risk drivers which we used to generate the dataset are the FX rate $\phi_t:=\phi^{\textrm{USDEUR}}_t$, the risk-neutral default intensity $\lambda_t$, and the counterparty default time $\tau$.

We simulate the FX rate via a Geometric Brownian Motion (GBM)
\begin{equation*}
	\frac{\diff{}\phi_t}{\phi_t} =\mu^{\bbP,\phi}\,\diff{}t +\sigma^{\bbP,\phi}\,\diff{}W^{\bbP,\phi}_t,
\end{equation*}
with $\mu^{\bbP,\phi}$ the drift, $\sigma^{\bbP,\phi}$ the volatility, and $W^{\bbP,\phi}_t$ a Wiener process.

Note that a risk neutral approach would imply the drift from the FX forwards market, and a volatility value from the FX option prices. On the other hand, here simulation is not aimed at pricing, but rather at generating training and testing data for our agents, and we need data showing a sufficient variety and richness so to mimic the time evolution of real data, which can significantly diverge from the risk neutral schemes. In this sense, we admit general constants for $\mu^{\bbP,\phi}$ and $\sigma^{\bbP,\phi}$. On the other hand, at a given time $t$, we suppose that the CVA pricing uses a  model $\bbQ_t$ with the same form as \eqref{eq:USDEUR_Q_dynamics}:
\begin{equation}\label{eq:USDEUR_Q_dynamics}
	\frac{\diff{}\phi_t}{\phi_t} = \mu^{\bbQ,\phi}\,\diff{}t + \sigma^{\bbQ,\phi}\, \diff{}W^{\bbQ,\phi}_t;
\end{equation}
note that for $\bbQ_t$ to be equivalent to the data generating measure $\bbP$ one would need that $\sigma^{\bbP,\phi} = \sigma^{\bbQ,\phi}$ by the Girsanov theorem, but we do not force this coherence as argued when $\bbQ_t$ was introduced in section~\ref{sec:CVA}.

We simulate the default intensity via the Cox Ingersoll Ross (CIR) model \citep{cox1985theory}
\begin{equation}\label{eq:lambda_P_dynamics}
	\diff{}\lambda_t = k^{\bbP,\lambda}(\theta^{\bbP,\lambda} - \lambda_t)\,\diff{}t + \sigma^{\bbP,\lambda} \sqrt{\lambda_t} \,\diff{}W^{\bbP,\lambda}_t,
\end{equation}
where $k^{\bbP,\lambda}$ is the mean reversion speed, $\theta^{\bbP,\lambda}$ is the long term intensity, $\sigma^{\bbP,\lambda}$ the volatility, and $W^{\bbP,\lambda}_t$ another Wiener process. Again we suppose that the CVA pricing is performed with a risk neutral model with the same form as \eqref{eq:lambda_P_dynamics}:
\begin{equation}\label{eq:lambda_Q_dynamics}
	\diff{}\lambda_t = k^{\bbQ,\lambda}(\theta^{\bbQ,\lambda} - \lambda_t)\,\diff{}t + \sigma^{\bbQ,\lambda} \sqrt{\lambda_t} \,\diff{}W^{\bbQ,\lambda}_t,
\end{equation}
where one can enforce equivalence of $\bbP$ and $\bbQ_t$ by forcing $\sigma^{\bbQ,\lambda}=\sigma^{\bbP,\lambda}$ and expressing the drift parameters of \eqref{eq:lambda_Q_dynamics} in terms of the parameters of \eqref{eq:lambda_P_dynamics} and a market price of risk proportional to either $\lambda_t^{1/2}$ or $\lambda_t^{-1/2}$, but doing so is not mandatory.

An instantaneous correlation $\rho^{\bbP}_{\lambda\phi}$ between the stochastic terms $\diff{}W^{\bbP,\phi}_t$ and $\diff{}W^{\bbP,\lambda}_t$ can be naturally introduced, which captures only mildly the relation between the jump to default and the FX term, but is sufficient to generate a correlation between the dynamics of the FX rate and the credit spread, whose effects on the hedging problem are among the main topics of this paper. We can also introduce a correlation term $\rho^{\bbQ}_{\lambda\phi}$ in the risk neutral pricing model $\bbQ_t$; for $\bbQ_t\sim\bbP$ one would need the two correlations to be equal, but we deem interesting also the case in which the pricing correlation is misspecified, e.g. $\rho^{\bbQ}_{\lambda\phi} = 0 \neq \rho^{\bbP}_{\lambda\phi}$.

Once $\phi_t$ is simulated, the bid and ask values for the FX rate are obtained as
\[
	\phi^{\textrm{USDEUR}, \pm}_t = \phi_t \pm \gamma^{\USD}
\] 
where we suppose that the semi-spread $\gamma_t^{\USD} = \gamma^{\USD}$ is a constant. Real bid-ask spreads are stochastic, but their short term movements are hard to model and generally uncorrelated with ordinary market moves, so we believe that using an average is enough to inform our agent of FX transaction costs. More relevant changes related to market regime switches might be relevant for longer term optimizations, and are left as a possible future development.

Analogously, once $\lambda_t$ is simulated, we generate bid and ask default intensities $\lambda^{\pm}$ by applying a semi-spread $\gamma^{\lambda}$:
\[
	\lambda^{\pm}_t = \lambda_t \pm \gamma^{\lambda},
\]
where again we believe that a constant $\gamma^{\lambda}$ with the correct order of magnitude is enough for our purposes. To fix such constant, we rely on the relation $\bar{\lambda} = s / (1-\Rec)$, valid for a deterministic flat default intensity $\bar{\lambda}$ where $s$ is the par spread of an idealized Credit Default Swap with recovery $\Rec$: inspired by that, we put
\[
	\gamma^{\lambda} = \frac{\gamma^{s}}{1-\Rec}
\]
with $\gamma^{s}$ an average bid-ask semi-spread for hedging CDS quotes, and $\Rec$ their conventional recovery defined in section~\ref{sec:CDS}. Eventually, standard pricing formulas for model \eqref{eq:lambda_Q_dynamics} \citep[e.g.][]{brigurio} map the simulated $\lambda^{\pm}_t$ to bid and ask prices $X_t^{\CDS_m,\pm}$; the mid price $X_t^{\CDS_m}$ and infinitesimal transaction cost $\gamma_t^{\CDS_m}$ are then defined by \eqref{eq:bidask}. 

The map from $(\phi_t, \lambda_t)$ to $X_t^{\CVA}$ is implemented as follows:
\begin{enumerate}
	\item If $\rho^{\bbQ}_{\lambda\phi} = 0$, by equations \eqref{eq:cva_indep}-\eqref{eq:cva_determ_rates}, where the expected value in \eqref{eq:cva_determ_rates} is given by a standard Black formula;
	\item If $\rho^{\bbQ}_{\lambda\phi} \neq 0$, by noting that \eqref{eq:cva_cox} implies via the Feynman-Kac theorem that $\CVA_t = \psi(t,\phi_t,\lambda_t)$ for a function $\psi$ satisfying the Partial Differential Equation 
	\begin{multline*}
		(1-\Rec)\lambda\max\left(0, \phi P_t^{c_1}(T)N^{c_1} -  P_t^{\bar{c}}(T)N^{\bar{c}}\right) = \biggl(\partial_t + \mu^{\bbQ,\phi}\phi\partial_\phi + k^{\bbQ,\lambda}(\theta^{\bbQ,\lambda} - \lambda)\partial_\lambda \\
		+ \frac12\left(\sigma^{\bbQ,\phi}\right)^2\phi^2\partial_{\phi\phi} + \frac12\left(\sigma^{\bbQ,\lambda}\right)^2\lambda\partial_{\lambda\lambda} + \rho^{\bbQ}_{\lambda\phi}\sigma^{\bbQ,\phi}\sigma^{\bbQ,\lambda}\phi\sqrt{\lambda}\partial_{\phi\lambda} - \left(\bar{c}+\lambda\right)\biggr)\psi
	\end{multline*}
	with terminal condition $\psi(T,\cdot,\cdot)=0$, which can be solved numerically.
\end{enumerate}

It remains to simulate default times. Let us suppose for simplicity that the real-world instantaneous Probability of Default (PD) $\bar{\lambda}_t$ is a deterministic multiple $\bar{m}(t)$ of the risk neutral intensity $\lambda_t$:
\begin{equation}\label{eq:intensity_mult}
	\bar{\lambda}_t = \bar{m}(t) \lambda_t,
\end{equation}
where $\bar{m}(t)$ can be used to fit e.g.~rating implied term PDs. Then the simplest implementation is to draw on each episode a standard exponential variable $\eps$ independent of all the Brownian drivers, and define
\[
	\tau = \bar{\Lambda}^{-1}(\eps) \qquad\text{where}\qquad \bar{\Lambda}(t) := \int_{t_0}^t \bar{\lambda}_s\,\diff{s}.
\]
However, unless the counterparty is very risky, this approach generates a random and often low number of paths where default is within the optimization horizon; so the agent has very scattered and noisy evidence to learn what happens in these scenarios, which however behave quite differently than survival ones, and can have a large impact on the PnL and its volatility.

Hence we propose importance sampling to compute the expected values under a law $\bbP'$ such that the default indicator is easily drawn from a distribution of choice before the drivers and default time. Specifically, we define $\bbP'$ to keep fixed both the law of the factors $(\phi,\lambda)$, and the law of $\tau$ given $I_{\{\tau \leq \bar{t}\}}$, while we change the conditional law of $I_{\{\tau \leq \bar{t}\}}$ given the risk factors
\[
	\bbP\left(\tau > \bar{t} \mid \phi,\lambda\right) := p(\phi,\lambda) = \exp\left(-\bar{\Lambda}(\bar{t})\right)
\]
to an independent Bernoulli draw
\[
	\bbP'\left(\tau > \bar{t} \mid \phi,\lambda\right) := p'.
\]
Now by a measure change, the expected value of an arbitrary $f(\phi,\lambda,\tau)$ can be rewritten as
\begin{align*}
	\bbE^\bbP\left[f(\phi,\lambda,\tau)\right] &= \bbE^{\bbP'}\left[f(\phi,\lambda,\tau)(p')^{-1}p(\phi,\lambda)I_{\{\tau>\bar{t}\}} + f(\phi,\lambda,\tau)(1-p')^{-1}(1-p(\phi,\lambda))I_{\{\tau\leq\bar{t}\}} \right]\\
	&= \bbE^{\bbP'(\cdot|\tau>\bar{t})}\left[f(\phi,\lambda,\tau)p(\phi,\lambda)\right] + \bbE^{\bbP'(\cdot|\tau\leq\bar{t})}\left[f(\phi,\lambda,\tau)(1-p(\phi,\lambda))\right]
\end{align*}
and each of the two laws appearing there is easy to simulate from, as the diffusion $(\phi,\lambda)$ is independent of the default indicator under $\bbP'$, while the conditional cumulative distribution functions of $\tau$ (equal by construction under $\bbP$ and $\bbP'$) are explicit:
\begin{equation*}
	\begin{cases}
		\bbP'\left(\tau > t \mid \tau > \bar{t},\lambda,\phi\right) = \exp\left(\bar{\Lambda}(\bar{t})-\bar{\Lambda}(t)\right),& t > \bar{t},\\
		\bbP'\left(\tau \leq t \mid \tau \leq \bar{t},\lambda,\phi\right) = \left(1-e^{-\bar{\Lambda}(\bar{t})}\right)^{-1}\left(1-e^{-\bar{\Lambda}(t)}\right),& t \leq \bar{t}.
	\end{cases}
\end{equation*}
Therefore, one can eventually estimate each of the two expected values separately by Monte Carlo on $B-B_0$ and $B_0$ paths respectively, where $B$ is the batch size and $B_0$ is a free hyperparameter.

\begin{remark}[Adaptive importance sampling]\label{rem:adaptive_importance_sampling}
Similar importance sampling ideas to handle rare events in RL have already been used by \citet{frank2008rare}; however, our scheme exploits the specifics of our problem in a few respects. Firstly, it does not change single step transitions: it only tweaks the law of a wisely chosen binary variable (the default indicator). Secondly, the proposal density is fixed, instead of being computed adaptively from estimates of the value function: this implies that adoption in the implementation of any RL algorithm is simpler. Finally, we are able to integrate out the law of $\{\tau \leq \bar{t}\}$, so that we can ensure that all batches have \emph{exactly} the same number of defaults: this should make learning more stable.
\end{remark}

\subsection{Reinforcement Learning setting}\label{sec:RL_impl}

In this subsection we specify the Markov Decision Process with the general notation of section~\ref{sec:RL_for_hedging}.

\subsubsection{State}\label{sec:state}
The state vector $\vc{S}_{t_i}$ includes the following real-valued quantities:
\begin{enumerate}
\item\label{it:time} Time $t$ to CVA maturity, in days.
\item\label{it:drivers} Value of the risk drivers $\lambda_t$ and $\phi_t$.
\item\label{it:allocation} Current allocation in the hedging assets, expressed without loss of generality by the first order sensitivity of the hedging book to $\lambda_t$ and $\phi_t$.
\item\label{it:CVA_sensy} Value of CVA, and its first order sensitivities with respect to $\lambda_t$ and $\phi_t$.
\item\label{it:CDS_sensy} Values of all hedging CDS, and their first order sensitivities with respect to $\lambda_t$.
\end{enumerate}

\begin{remark}[Sensitivities]\label{rem:sensy}
Items \ref{it:CVA_sensy} and \ref{it:CDS_sensy} are functions of $(t, \lambda_t, \phi_t)$ which the algorithm may learn by itself, but every modern front office system already calculates first order risks of all book items, so there is no reason not to give them to the RL agent as useful pre-engineered features. A similar consideration applies to item \ref{it:allocation}: one may naively put into the state the notionals $N^{\CDS}$ and $N^{f(\USD)}$, but sensitivities are universally considered by practitioners as a better representation of risk when trading, so we let the AI trader start from information in this form.
\end{remark}

\subsubsection{Action}\label{sec:action}

As we said in section~\ref{sec:RL_for_hedging}, the action vector $\vc{A}_{t_i}$ should describe the allocations in the interval $(t_i, t_{i+1}]$. For the same reasons as in remark~\ref{rem:sensy}, we express them as a vector of two sensitivities to $\lambda_t$ and $\phi_t$.

\subsubsection{Reward}\label{sec:reward}
By combining the definition of reward \eqref{eq:reward_as_PnL} and that of gain \eqref{eq:gain}, we get
\begin{multline}\label{eq:reward}
    R_{t_{i+1}} \approx \sum_{h\in\mathscr{H}} N_{t_{i+1}}^h (\bar{X}_{t_{i+1}}^h - \bar{X}_{t_i}^h + \bar{Y}_{t_{i+1}}^h - \bar{Y}_{t_i}^h) - \sum_{m\in\mathscr{M}} \gamma_{t_i}^{\CDS_m} \left|N_{t_{i+1}}^{\CDS_m} - N_{t_i}^{\CDS_m}\right| + \\
    - \gamma_{t_i}^{\USD}\left[\left|r^{\USD}N_{t_{i+1}}^{f(\USD)}\right|(t_{i+1}-t_{i}) + \left|N_{t_{i+1}}^{f(\USD)} - N_{t_i}^{f(\USD)}\right| \right],
\end{multline}
where the first summation contributes to the reward with the effects of price variations and dividends, while the remaining addends correspond to the costs arising from portfolio rebalances.

\begin{remark}[Discrete increments]
In computing $R_{t_{i+1}}$, the assumption that $N^{h}$ is piecewise constant for $h\in\mathscr{F}$ would allow to replace integrals with full-step finite increments only when the integrand is $N^{\CDS_m}$ or $N^{f(\USD)}$, but we did it for simplicity also for the other integrals: this is a quite accurate approximation because $t_{i+1}-t_{i}$ is small, and in any case, any desired precision can be achieved by a straightforward generalized implementation using finer Riemann sums. On the other hand, if interest rates are null then \emph{all} integrands are piecewise constant up to default time, and equation \eqref{eq:reward} is exact except for the tiny effect of discretizing $\tau$ on the simulation grid.
\end{remark}

\section{Numerical results}\label{sec:numerics}

We collect here empirical evidence on the behaviour of the algorithm, and of the optimized policy with different choices of model parameters and algorithmic hyperparameters.

\subsection{Common parameters}\label{sec:parameters}

\begin{table}
	\centering
	\begin{tabular}{llr}
		\toprule
		Description & Symbol & Value\\
		\colrule
		Time to maturity of the FX forward & $T-t_0$ & 5.0\\
		Dollars exchanged at maturity & $N^{\USD}$ & 1.1\\
		Euros exchanged at maturity & $N^{\EUR}$ & 1.0\\
		Recovery fraction of CVA and of the CDS & $\Rec$ & 40\%\\
		Value of EUR interest rates & $r^{f(\EUR)}=r^{a(\CDS_m)}$ & 3.3\%\\
		Value of USD interest rates & $r^{f(\USD)}$ & 4.5\%\\
		Mid exchange rate USDEUR at the pricing date & $\phi_{t_0}$ & 1.0\\
		USDEUR trading cost parameter & $\gamma^{\USD}$ & $5\times10^{-5}$\\
		Risk-neutral drift of the USDEUR FX rate & $\mu^{\bbQ,\phi}$ & -1.2\%\\
		Volatility of the USDEUR FX rate & $\sigma^{\bbP,\phi}=\sigma^{\bbQ,\phi}$ & 10\%\\
		Risk-neutral correlation of Brownian drivers & $\rho^{\bbQ}_{\lambda\phi}$ & 0\%\\
		\botrule
	\end{tabular}
	\caption{Base value of termsheet and model parameters. Time unit is years.}\label{tab:model_pars}
\end{table}

The financial parameters in table \ref{tab:model_pars} are kept fixed throughout all experiments. The time grid $t_0<t_1<\dots<t_N$ spans 90 trading days and considers a 2-hours spacing within each trading day, for a total of 5 timesteps per day; note that this realistically implies a non-uniform spacing in calendar time, with larger steps (and market movements) across the nights, and even larger no-action gaps due to weekends. The actor and critic in the TRVO algorithm are represented by a neural network with two hidden layers of 10 units each and hyperbolic tangent activation function; and trained with batches of 500 episodes.

Performance metrics are computed with $\gamma=1$ even though $\gamma=0.95$ is used in training to ease convergence.

\subsection{Runs with negligible defaults}\label{sec:run_no_defaults}

\begin{table}
	\centering
	\begin{tabular}{llr}
		\toprule
		Description & Symbol & Value\\
		\colrule
		Instantaneous default intensity at the pricing date & $\lambda_{t_0}$ & 1.66\%\\
		Mean reversion speed of the default intensity & $k^{\bbP,\lambda}=k^{\bbQ,\lambda}$ & 0.3769\\
		Long term mean of the default intensity & $\theta^{\bbP,\lambda}=\theta^{\bbQ,\lambda}$ & 1.87\%\\
		CIR volatility coefficient of the default intensity & $\sigma^{\bbP,\lambda}=\sigma^{\bbQ,\lambda}$ &  19.22\%\\
		CDS trading cost parameter & $\gamma^{\lambda}$ & $8.3\times10^{-4}$\\
		Real-world correlation of Brownian drivers & $\rho^{\bbP}_{\lambda\phi}$ & 50\%\\
		Real-world drift of the USDEUR FX rate & $\mu^{\bbP,\phi}$ & 0\%\\
		\botrule
	\end{tabular}
	\caption{Value of termsheet and model parameters for section~\ref{sec:run_no_defaults}. Credit parameters represent a counterparty with an initial flat CDS spread term structure of 100 bps with a bid-ask semi-spread of about 5 bps. Time unit is years.}\label{tab:negligible_defaults_model_pars}
\end{table}

\begin{figure}
    \centering
    \begin{tikzpicture}

    \begin{axis}
        [
        	table/col sep=semicolon,
            height=7cm, width=\linewidth,
            title={Return-volatility comparison},
            xmajorgrids=true,
            ymajorgrids=true,
            xlabel={$\hat{\nu}^2_\pi$},
            xmin=0, xmax=0.00000018,
            x tick label style={
                /pgf/number format/fixed,
                /pgf/number format/fixed zerofill,
                /pgf/number format/precision=1,
            },
            ylabel={$\hat{J}_\pi$},
            ymin=-0.0012, ymax=-0.0003,
            y tick label style={
                /pgf/number format/fixed,
                /pgf/number format/fixed zerofill,
                /pgf/number format/precision=1,
            },
        ]
        \addplot[
            AgentStyle,
            skip coords between index={0}{1},
            visualization depends on={value \thisrow{anchor}\as\myanchor},
        ] table[x=x, y=y, meta=meta] {./nodefault_100bps_performance.csv};
        \addplot+[
            BaselineStyle,
            skip coords between index={1}{9},
            visualization depends on={value \thisrow{anchor}\as\myanchor},
        ] table[x=x, y=y, meta=meta] {./nodefault_100bps_performance.csv};
        \draw [dashed] (1.25510910360427E-08, -0.00106650290974211) -- (1.25510910360427E-08, -0.000451118588840289);
    \end{axis}
    
\end{tikzpicture}
    \caption{Each dot represents the average performance, depending on $\beta$ annotated next to each dot, of an agent over 2,000 out-of-sample episodes in terms of return and unnormalized reward volatility, for a counterparty whose $\bbP$ probability of default is infinitesimal.  Model parameters are specified in table~\ref{tab:negligible_defaults_model_pars}, and CVA at inception equals $-3.34\times10^{-3}$ EUR.}
    \label{fig:no_defaults_pareto}
\end{figure}
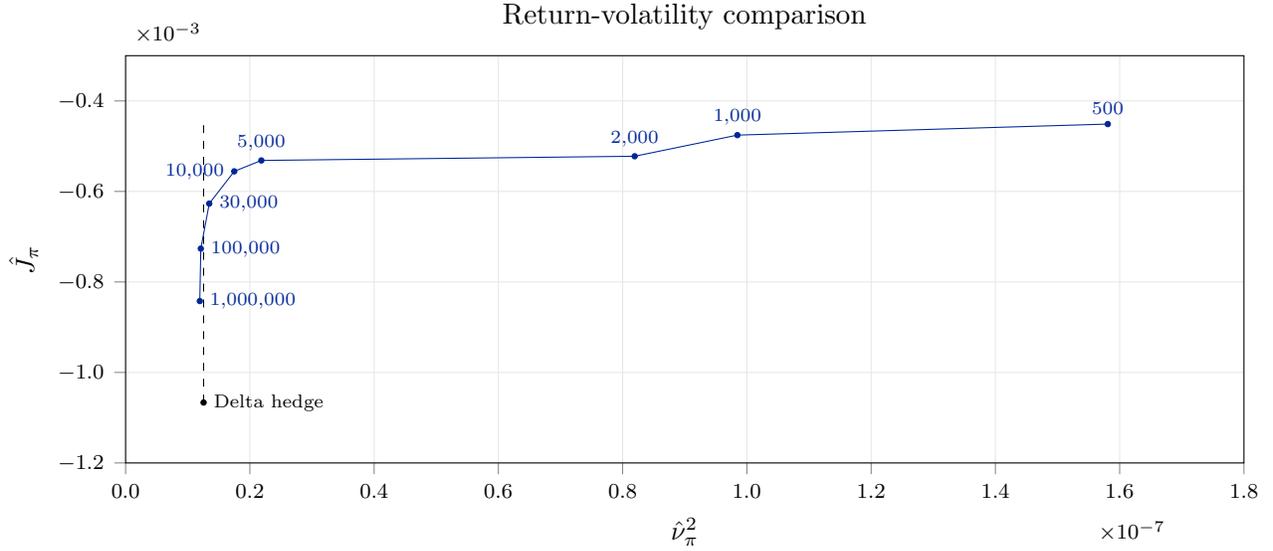
\begin{figure}
    \centering
       \begin{tikzpicture}

    \begin{axis}
        [
            height=6cm, width=\linewidth,
            xmajorgrids=true,
            xmin=0, xmax=450,
            xlabel={Timestep},
            ylabel={CDS Action},
			yticklabel style={
		        /pgf/number format/fixed,
		        /pgf/number format/fixed zerofill,
        		/pgf/number format/precision=3
			},
			scaled y ticks=false,
            legend style={
                legend columns=5,
                at={(0.5, -0.25)},
                anchor=north
            }
        ]
        \addplot[LineStyle, black] table[y=delta_crd] {./nodefault_100bps_ep1-actions.csv}; \addlegendentry{Delta hedge}
        \addplot[LineStyle, dashdotted, txtOrange] table[y=5e2_crd] {./nodefault_100bps_ep1-actions.csv}; \addlegendentry{$\beta=500$}
        \addplot[LineStyle, dashdotted, txtRed3] table[y=1e3_crd] {./nodefault_100bps_ep1-actions.csv}; \addlegendentry{$\beta=\text{1,000}$}
        \addplot[LineStyle, dashdotted, txtRed] table[y=2e3_crd] {./nodefault_100bps_ep1-actions.csv}; \addlegendentry{$\beta=\text{2,000}$}
        \addplot[LineStyle, dashdotted, txtGreen] table[y=5e3_crd] {./nodefault_100bps_ep1-actions.csv}; \addlegendentry{$\beta=\text{5,000}$}
        \addplot[LineStyle, dashdotted, txtGreen3] table[y=1e4_crd] {./nodefault_100bps_ep1-actions.csv}; \addlegendentry{$\beta=\text{10,000}$}
        \addplot[LineStyle, dashdotted, txtBlue3] table[y=3e4_crd] {./nodefault_100bps_ep1-actions.csv}; \addlegendentry{$\beta=\text{30,000}$}
        \addplot[LineStyle, dashdotted, txtBlue] table[y=1e5_crd] {./nodefault_100bps_ep1-actions.csv}; \addlegendentry{$\beta=\text{100,000}$}
        \addplot[LineStyle, dashdotted, txtPurple] table[y=1e6_crd] {./nodefault_100bps_ep1-actions.csv}; \addlegendentry{$\beta=\text{1,000,000}$}
    \end{axis}

    \begin{axis}
        [
            height=6cm, width=\linewidth,
            yshift=4.5cm,
            title={Evolution of the action over time},
            xmajorgrids=true,
            xtick style={draw=none},
            xticklabels={,,},
            xmin=0, xmax=450,
            ylabel={FX Action},
			yticklabel style={
		        /pgf/number format/fixed,
		        /pgf/number format/fixed zerofill,
        		/pgf/number format/precision=3
			},
			scaled y ticks=false,
        ]
        \addplot[LineStyle, black] table[y=delta_fx] {./nodefault_100bps_ep1-actions.csv};
        \addplot[LineStyle, dashdotted, txtOrange] table[y=5e2_fx] {./nodefault_100bps_ep1-actions.csv};
        \addplot[LineStyle, dashdotted, txtRed3] table[y=1e3_fx] {./nodefault_100bps_ep1-actions.csv};
        \addplot[LineStyle, dashdotted, txtRed] table[y=2e3_fx] {./nodefault_100bps_ep1-actions.csv};
        \addplot[LineStyle, dashdotted, txtGreen] table[y=5e3_fx] {./nodefault_100bps_ep1-actions.csv};
        \addplot[LineStyle, dashdotted, txtGreen3] table[y=1e4_fx] {./nodefault_100bps_ep1-actions.csv};
        \addplot[LineStyle, dashdotted, txtBlue3] table[y=3e4_fx] {./nodefault_100bps_ep1-actions.csv};
        \addplot[LineStyle, dashdotted, txtBlue] table[y=1e5_fx] {./nodefault_100bps_ep1-actions.csv};
        \addplot[LineStyle, dashdotted, txtPurple] table[y=1e6_fx] {./nodefault_100bps_ep1-actions.csv};
    \end{axis}

\end{tikzpicture}
    \caption{The plot represents the FX action (top) and CDS action (bottom) for an out-of-sample episode, expressed as the sensitivity to respectively $\smash{\phi_t}$ and $\smash{\lambda_t}$ of the hedging portfolio, chosen by the delta hedging strategy and by agents trained at different values of the risk aversion coefficient $\beta$, for a counterparty whose $\bbP$ probability of default is infinitesimal. Model parameters are specified in table~\ref{tab:negligible_defaults_model_pars}.}
    \label{fig:negligible_defaults_actions}
\end{figure}
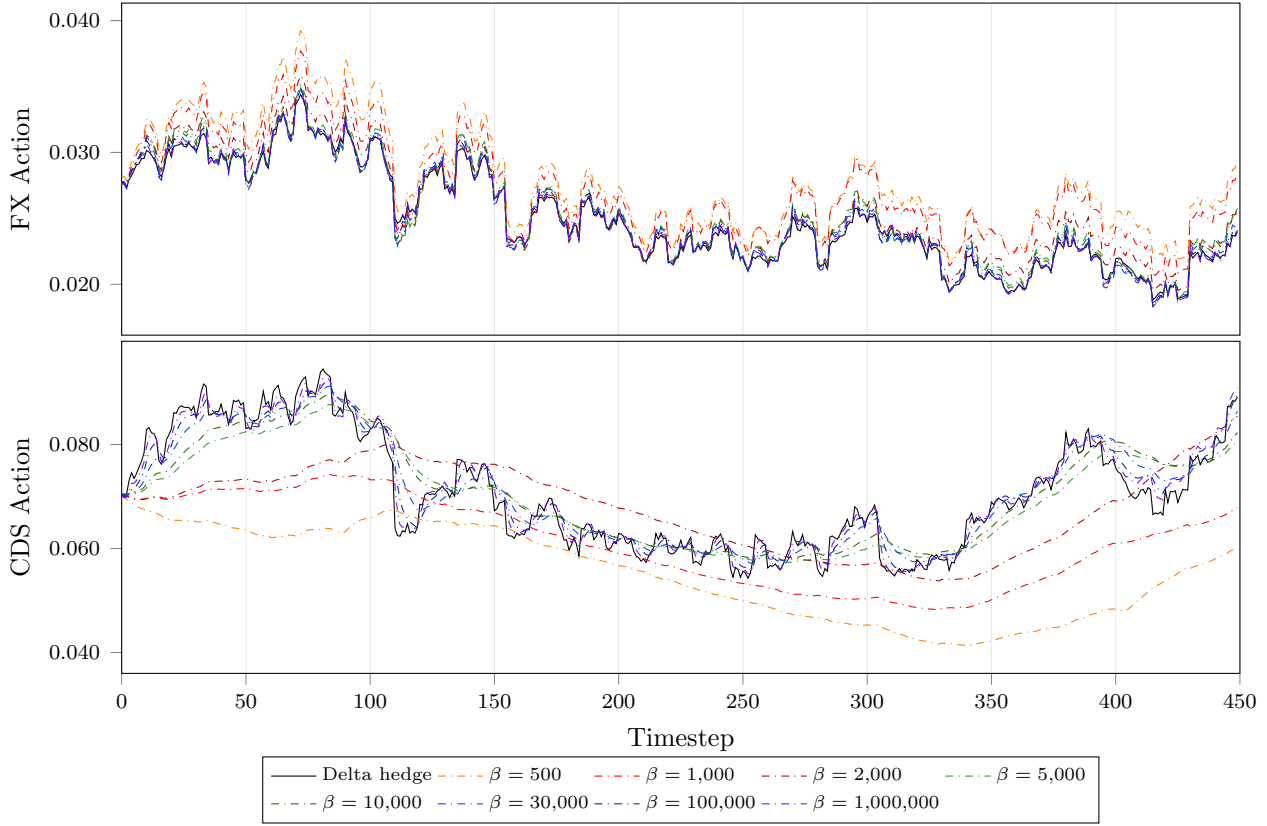
\begin{figure}
    \centering
    \begin{tikzpicture}
   \pgfplotsset{
        ScatterStyle/.style={
            mark options={solid, scale=0.5},
            only marks,
            mark=*
        }
    }

    \begin{axis}
        [
            height=7cm, width=\linewidth,
            title={Agent vs delta hedging actions},
            xmajorgrids=true,
            ymajorgrids=true,
            xlabel={FX Action minus FX Delta},
            x tick label style={
                /pgf/number format/fixed,
                /pgf/number format/fixed zerofill,
                /pgf/number format/precision=1,
            },
            ylabel={CDS Action minus CDS Delta},
            scaled y ticks=false,
            y tick label style={
                /pgf/number format/fixed,
                /pgf/number format/fixed zerofill,
                /pgf/number format/precision=2
            },
            legend style={
                legend columns=1,
                legend pos=north east,
            }
        ]        
        \addplot[ScatterStyle, txtOrange] table[x=5e2_fx, y=5e2_crd] {./nodefault_100bps_ep1-scatter.csv}; \addlegendentry{$\beta=500$}
        \addplot[ScatterStyle, txtRed3] table[x=1e3_fx, y=1e3_crd] {./nodefault_100bps_ep1-scatter.csv}; \addlegendentry{$\beta=\text{1,000}$}
        \addplot[ScatterStyle, txtRed] table[x=2e3_fx, y=2e3_crd] {./nodefault_100bps_ep1-scatter.csv}; \addlegendentry{$\beta=\text{2,000}$}
        \addplot[ScatterStyle, txtGreen] table[x=5e3_fx, y=5e3_crd] {./nodefault_100bps_ep1-scatter.csv}; \addlegendentry{$\beta=\text{5,000}$}
        \addplot[ScatterStyle, txtGreen3] table[x=1e4_fx, y=1e4_crd] {./nodefault_100bps_ep1-scatter.csv}; \addlegendentry{$\beta=\text{10,000}$}
        \addplot[ScatterStyle, txtBlue3] table[x=3e4_fx, y=3e4_crd] {./nodefault_100bps_ep1-scatter.csv}; \addlegendentry{$\beta=\text{30,000}$}
        \addplot[ScatterStyle, txtBlue] table[x=1e5_fx, y=1e5_crd] {./nodefault_100bps_ep1-scatter.csv}; \addlegendentry{$\beta=\text{100,000}$}
        \addplot[ScatterStyle, txtPurple] table[x=1e6_fx, y=1e6_crd] {./nodefault_100bps_ep1-scatter.csv}; \addlegendentry{$\beta=\text{1,000,000}$}

        \draw [dashed, gray] (0, -0.05) -- (0, 0.03);
        \draw [dashed, gray] (-0.003, 0) -- (0.006, 0);
    \end{axis}
    
\end{tikzpicture}
    \caption{Distribution on an out-of-sample path of the difference, depending on $\beta$, between the agent and the delta hedge in terms of FX action and CDS action, expressed respectively as the sensitivity to $\smash{\phi_t}$ and $\smash{\lambda_t}$ of the hedging portfolio, for a counterparty whose $\bbP$ probability of default is infinitesimal. Model parameters are specified in table~\ref{tab:negligible_defaults_model_pars}.}
    \label{fig:no_defaults_scatter}
\end{figure}
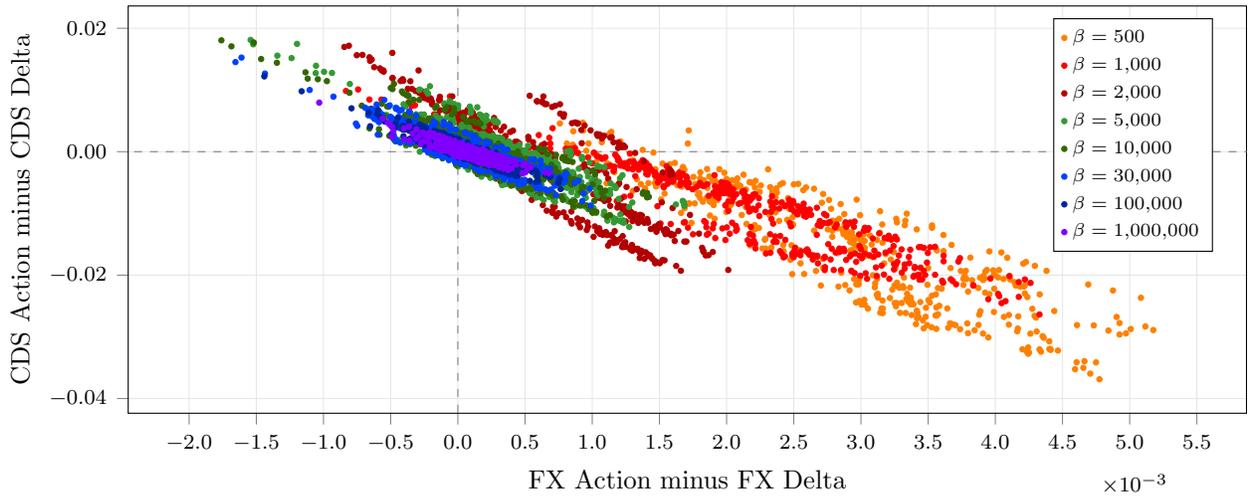
\begin{figure}
    \centering
       \begin{tikzpicture}

    \begin{axis}
        [
            height=6cm, width=\linewidth,
            xmajorgrids=true,
            ymajorgrids=true,
            xlabel={CDS Action minus CDS Delta},
            ymin=0,
            ylabel={Density},
            x tick label style={
                /pgf/number format/fixed,
                /pgf/number format/fixed zerofill,
                /pgf/number format/precision=1,
            },
            legend style={
                legend columns=4,
                at={(0.5, -0.25)},
                anchor=north,
            }
        ]
        \addplot[LineStyle, txtOrange] table[x=5e2_x_crd, y=5e2_y_crd] {./nodefault_100bps_ep1-kde.csv}; \addlegendentry{$\beta=500$}
        \addplot[LineStyle, txtRed3] table[x=1e3_x_crd, y=1e3_y_crd] {./nodefault_100bps_ep1-kde.csv}; \addlegendentry{$\beta=\text{1,000}$}
        \addplot[LineStyle, txtRed] table[x=2e3_x_crd, y=2e3_y_crd] {./nodefault_100bps_ep1-kde.csv}; \addlegendentry{$\beta=\text{2,000}$}
        \addplot[LineStyle, txtGreen] table[x=5e3_x_crd, y=5e3_y_crd] {./nodefault_100bps_ep1-kde.csv}; \addlegendentry{$\beta=\text{5,000}$}
        \addplot[LineStyle, txtGreen3] table[x=1e4_x_crd, y=1e4_y_crd] {./nodefault_100bps_ep1-kde.csv}; \addlegendentry{$\beta=\text{10,000}$}
        \addplot[LineStyle, txtBlue3] table[x=3e4_x_crd, y=3e4_y_crd] {./nodefault_100bps_ep1-kde.csv}; \addlegendentry{$\beta=\text{30,000}$}
        \addplot[LineStyle, txtBlue] table[x=1e5_x_crd, y=1e5_y_crd] {./nodefault_100bps_ep1-kde.csv}; \addlegendentry{$\beta=\text{100,000}$}
        \addplot[LineStyle, txtPurple] table[x=1e6_x_crd, y=1e6_y_crd] {./nodefault_100bps_ep1-kde.csv}; \addlegendentry{$\beta=\text{1,000,000}$}
        
		\draw [dashed, txtOrange] (-0.015363616621206774, 0) -- (-0.015363616621206774, 30.518661990289882);
		\draw [dashed, txtRed3] (-0.00944055391552075, 0) -- (-0.00944055391552075, 44.14389702656063);
		\draw [dashed, txtRed] (-0.0026125177699170866, 0) -- (-0.0026125177699170866, 32.116171341378646);
		\draw [dashed, txtGreen] (-0.0008899532154558185, 0) -- (-0.0008899532154558185, 72.79834961745503);
		\draw [dashed, txtGreen3] (-4.375847766276276e-05, 0) -- (-4.375847766276276e-05, 87.5192508417951);
		\draw [dashed, txtBlue3] (0.00026554970275927814, 0) -- (0.00026554970275927814, 114.15864037172533);
		\draw [dashed, txtBlue] (0.00029315710237581805, 0) -- (0.00029315710237581805, 181.39858176408427);
		\draw [dashed, txtPurple] (0.0003133113938889929, 0) -- (0.0003133113938889929, 313.01946012393745);
		
    \end{axis}

    \begin{axis}
        [
            height=6cm, width=\linewidth,
            yshift=5.7cm,
            xmajorgrids=true,
            ymajorgrids=true,
            x tick label style={
                /pgf/number format/fixed,
                /pgf/number format/fixed zerofill,
                /pgf/number format/precision=0,
            },
            xlabel={FX Action minus FX Delta},
            ymin=0,
            ylabel={Density},
            scaled y ticks=false,
        ]
        \addplot[LineStyle, txtOrange] table[x=5e2_x_fx, y=5e2_y_fx] {./nodefault_100bps_ep1-kde.csv};
        \addplot[LineStyle, txtRed3] table[x=1e3_x_fx, y=1e3_y_fx] {./nodefault_100bps_ep1-kde.csv};
        \addplot[LineStyle, txtRed] table[x=2e3_x_fx, y=2e3_y_fx] {./nodefault_100bps_ep1-kde.csv};
        \addplot[LineStyle, txtGreen] table[x=5e3_x_fx, y=5e3_y_fx] {./nodefault_100bps_ep1-kde.csv};
        \addplot[LineStyle, txtGreen3] table[x=1e4_x_fx, y=1e4_y_fx] {./nodefault_100bps_ep1-kde.csv};
        \addplot[LineStyle, txtBlue3] table[x=3e4_x_fx, y=3e4_y_fx] {./nodefault_100bps_ep1-kde.csv};
        \addplot[LineStyle, txtBlue] table[x=1e5_x_fx, y=1e5_y_fx] {./nodefault_100bps_ep1-kde.csv};
        \addplot[LineStyle, txtPurple] table[x=1e6_x_fx, y=1e6_y_fx] {./nodefault_100bps_ep1-kde.csv};

		\draw [dashed, txtOrange] (0.0028762194365079914, 0) -- (0.0028762194365079914, 355.9922128219893);
		\draw [dashed, txtRed3] (0.00213784573531727, 0) -- (0.00213784573531727, 402.2504053005875);
		\draw [dashed, txtRed] (0.0006919359795501536, 0) -- (0.0006919359795501536, 498.04526055510576);
		\draw [dashed, txtGreen] (0.0004274919057412747, 0) -- (0.0004274919057412747, 868.0615554165664);
		\draw [dashed, txtGreen3] (0.00025271234483964056, 0) -- (0.00025271234483964056, 910.288541475624);
		\draw [dashed, txtBlue3] (0.0001137316560776538, 0) -- (0.0001137316560776538, 934.1068341172545);
		\draw [dashed, txtBlue] (7.482335963429473e-05, 0) -- (7.482335963429473e-05, 1494.6640028005907);
		\draw [dashed, txtPurple] (5.829361289526948e-05, 0) -- (5.829361289526948e-05, 2077.930604549693);
		
	\end{axis}

\end{tikzpicture}
    \caption{Kernel density estimators on an out-sample-path of the distribution of the agent's additive corrections with respect to the delta hedging benchmark, for different values of the risk aversion parameter $\beta$, for a counterparty whose $\bbP$ probability of default is infinitesimal. Model parameters are specified in table~\ref{tab:negligible_defaults_model_pars}.}
    \label{fig:no_defaults_kde}
\end{figure}
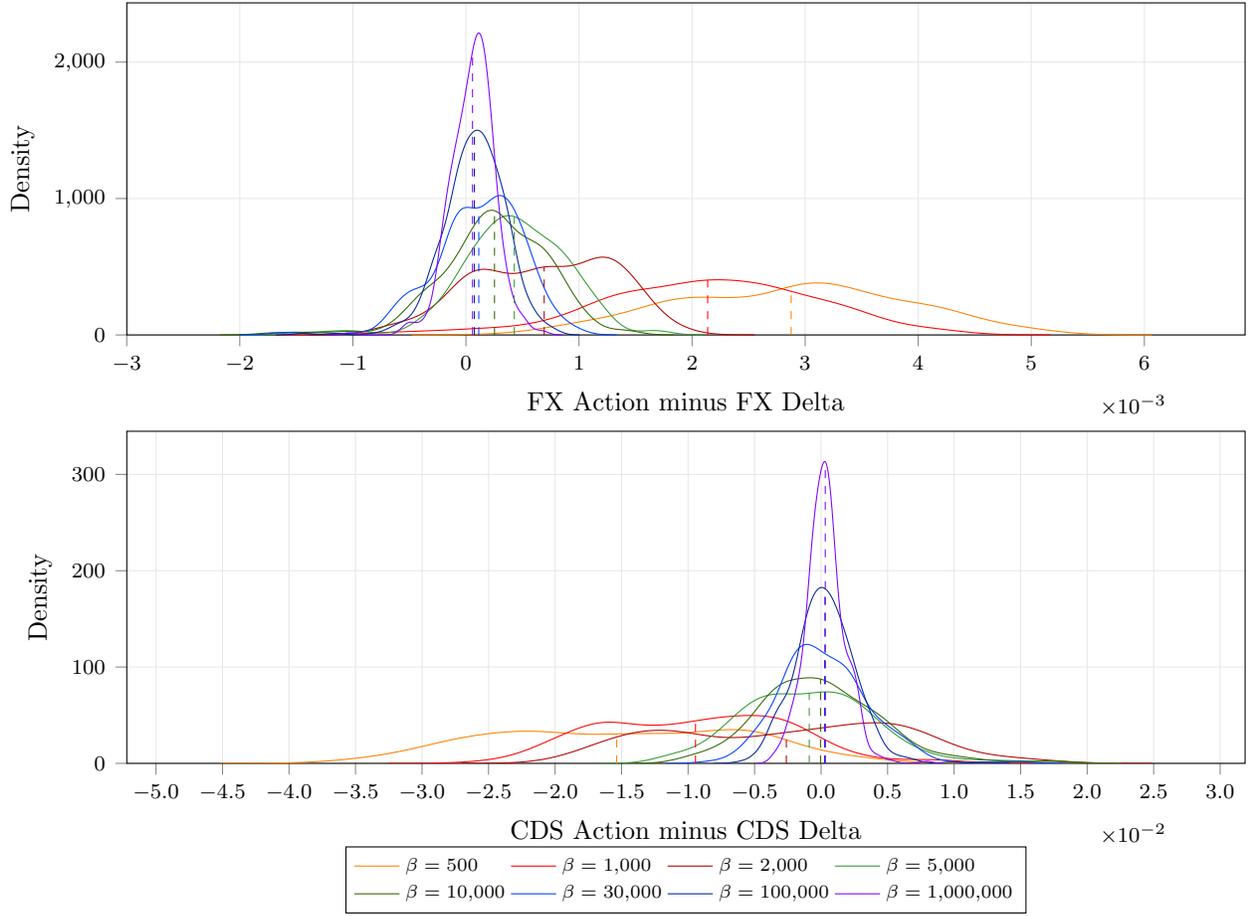

In this subsection, we reproduce a common practical situation in that:
\begin{enumerate}
\item\label{it:no_default} The multiplier defined in equation \eqref{eq:intensity_mult} satisfies $\bar{m}\ll 1$, so that the real-world probability of a default event is assumed negligible.
\item\label{it:different_rho} $\rho^{\bbP}_{\lambda\phi} \neq \rho^{\bbQ}_{\lambda\phi}$, so that CVA pricing is ignoring an existing correlation; precisely, $\rho^{\bbP}_{\lambda\phi}=50\%$.
\item\label{it:different_mu} $\mu^{\bbP,\phi}\neq \mu^{\bbQ,\phi}$, so that the FX drift is not the difference of rates; precisely, we take $\mu^{\bbP,\phi}=0\%$.
\end{enumerate}
Since the jump-to-default risk is neglected in this setup, one hedging CDS should be enough: we assume that its maturity equals that of the FX forward. Model parameters are specified in table~\ref{tab:negligible_defaults_model_pars}.

The results are plotted in figure~\ref{fig:no_defaults_pareto} in terms of the two competing objectives $\hat{J}_{\pi}$ and $\hat{\nu}_{\pi}$. In spite of the higher complexity of the environment, we get the same qualitative findings of \citet{daluiso2023Acva}, namely that the optimized agent's reward volatility is comparable to that of the delta hedging benchmark for high risk aversions, but that at all levels of $\beta$, the agent has lower costs than the delta hedging benchmark. The two main observations are reproduced:
\begin{itemize}
\item The agent approximately tracks the sensitivities of CVA to the risk factors $\lambda_t$ and $\phi_t$, but the allocations are smoother in the time direction, to avoid frequent small rebalances with opposite sign, as one can appreciate from the sample paths in figure~\ref{fig:negligible_defaults_actions}.
\item The rebalances of the more costly CDS instrument are reduced even further by partly cross-hedging the exposure to $\lambda$ via its correlation to $\phi$, i.e.~on average the agent overhedges the FX sensitivity when it is under-hedged in credit sensitivity, as one can see from the scatter plot of the corrections in figure~\ref{fig:no_defaults_scatter}.
\end{itemize}
Moreover, we can make the following new remarks on the effect of the distinction of $\bbP$ from $\bbQ$:
\begin{itemize}
\item The allocation in the USD bank account is slightly upward-biased with respect to the delta hedge, as it takes some risk to profit from the null $\bbP$ drift of the exchange rate, which makes a long position on USDEUR profitable (as USD rates are higher than EUR rates): see the distribution of FX actions in the upper plot of figure~\ref{fig:no_defaults_kde}.
\item The allocation in the protection-buyer CDS is slightly downward-biased with respect to the delta hedge, as it takes some market risk to profit from the fact that the counterparty never defaults: see the distribution of CDS actions in the lower plot of figure~\ref{fig:no_defaults_kde}.
\end{itemize}
All of the above effects are more pronounced for lower values of the risk aversion coefficient $\beta$. 

\subsection{Runs with frequent defaults}\label{sec:run_many_defaults}

\begin{table}
	\centering
	\begin{tabular}{llr}
		\toprule
		Description & Symbol & Value\\
		\colrule
		Instantaneous default intensity at the pricing date & $\lambda_{t_0}$ & 7.824\%\\
		Mean reversion speed of the default intensity & $k^{\bbP,\lambda}=k^{\bbQ,\lambda}$ & 0.19483\\
		Long term mean of the default intensity & $\theta^{\bbP,\lambda}=\theta^{\bbQ,\lambda}$ & 15.26\%\\
		CIR volatility coefficient of the default intensity & $\sigma^{\bbP,\lambda}=\sigma^{\bbQ,\lambda}$ & 35.973\%\\
		CDS trading cost parameter & $\gamma^{\lambda}$ & $1.66\times10^{-3}$\\
		Real-world correlation of Brownian drivers & $\rho^{\bbP}_{\lambda\phi}$ & 0\%\\
		Real-world drift of the USDEUR FX rate & $\mu^{\bbP,\phi}$ & -1.2\%\\
		\botrule
	\end{tabular}
	\caption{Value of termsheet and model parameters for section~\ref{sec:run_many_defaults}. Credit parameters represent a counterparty with an initial flat CDS spread term structure of 500 bps with a bid-ask semi-spread of about 10 bps. Time unit is years.}\label{tab:many_defaults_model_pars}
\end{table}

\begin{figure}
    \centering
    \begin{tikzpicture}

    \begin{axis}
        [
        	table/col sep=semicolon,
            height=6cm, width=\linewidth,
            title={Return-volatility comparison},
            xmajorgrids=true,
            ymajorgrids=true,
            xlabel={$\hat{\nu}^2_\pi$},
            x tick label style={
                /pgf/number format/fixed,
                /pgf/number format/fixed zerofill,
                /pgf/number format/precision=1,
            },
            ylabel={$\hat{J}_\pi$},
            y tick label style={
                /pgf/number format/fixed,
                /pgf/number format/fixed zerofill,
                /pgf/number format/precision=1,
            },
        ]
        \addplot[
	        BaselineStyle,
        	skip coords between index={2}{12},
            visualization depends on={value \thisrow{anchor}\as\myanchor},
        ] table[x=x, y=y, meta=meta] {./default_500bps_1cds_performance.csv};
        \addplot+[
            AgentStyle,
            skip coords between index={0}{2},
            visualization depends on={value \thisrow{anchor}\as\myanchor},
        ] table[x=x, y=y, meta=meta] {./default_500bps_1cds_performance.csv};
    \end{axis}
    
\end{tikzpicture}
    \caption{Each dot represents the average performance, depending on $\beta$ annotated next to each dot, of an agent over 2,000 out-of-sample episodes in terms of return and unnormalized reward volatility, for a counterparty whose $\bbP$ probability of default is relevant. In all simulations $\mathcal{M}=\{\text{5Y}\}$, and CVA at inception equals $-1.40\times10^{-2}$ EUR.  Model parameters are specified in table~\ref{tab:many_defaults_model_pars}.}
    \label{fig:many_defaults_1cds_pareto}
\end{figure}
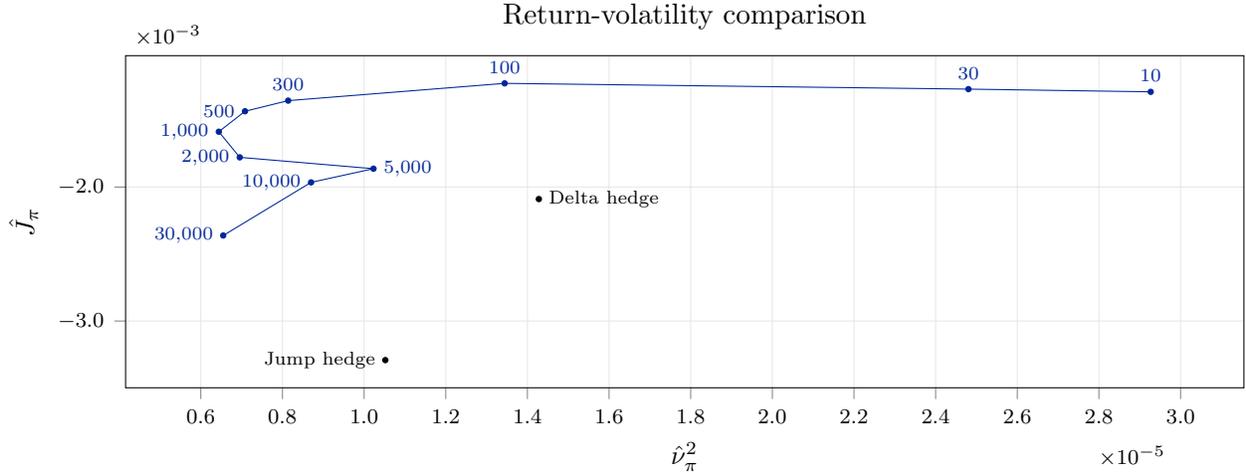
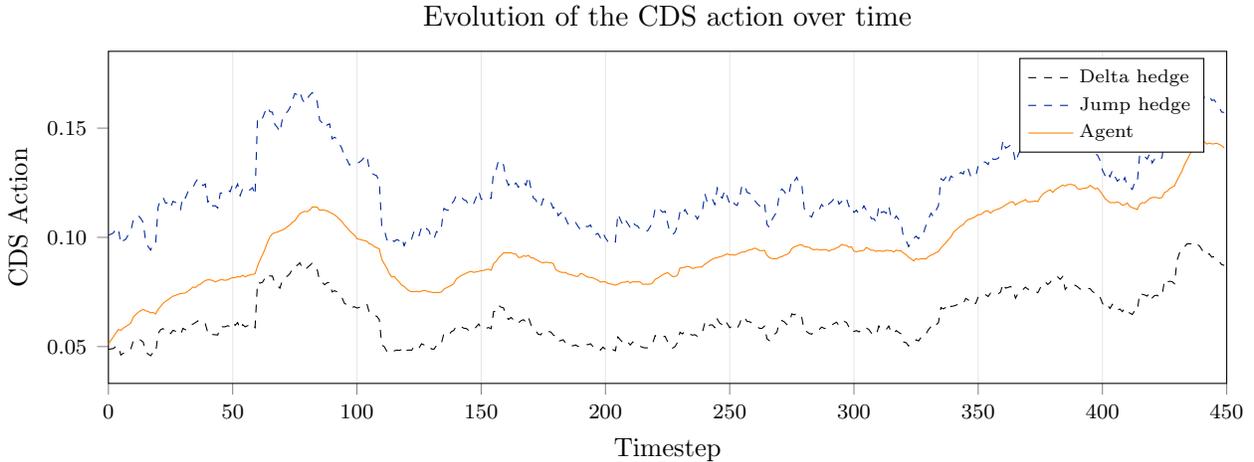
\begin{figure}
    \centering
    \begin{tikzpicture}

    \begin{axis}
        [
            height=6cm, width=\linewidth,
            title={Evolution of the CDS action over time},
            xlabel={Timestep},
            xmajorgrids=true,
            xmin=0, xmax=450,
            ylabel={CDS Action},
            y tick label style={
                /pgf/number format/fixed,
                /pgf/number format/fixed zerofill,
                /pgf/number format/precision=2,
            },
            legend style={
                legend columns=1,
            }
        ]
        \addplot[LineStyle, black, dashed] table[y=delta] {./default_500bps_1cds_ep1-cr-actions.csv}; \addlegendentry{Delta hedge}
        \addplot[LineStyle, txtBlue, dashed] table[y=jump] {./default_500bps_1cds_ep1-cr-actions.csv}; \addlegendentry{Jump hedge}
        \addplot[LineStyle, txtOrange] table[y=500] {./default_500bps_1cds_ep1-cr-actions.csv}; \addlegendentry{Agent}
    \end{axis}

\end{tikzpicture}
    \caption{Single CDS action chosen by the agent and the delta hedging and jump hedging benchmarks on an out-of-sample episode, all expressed as the sensitivity to $\smash{\lambda_t}$ of the hedging portfolio, for a counterparty whose $\bbP$ probability of default is relevant. Model parameters are specified in table~\ref{tab:many_defaults_model_pars}.}
    \label{fig:many_defaults_1cds_actions}
\end{figure}
\begin{figure}
    \centering
    \begin{tikzpicture}

    \begin{axis}
        [
        	table/col sep=semicolon,
            height=7cm, width=\linewidth,
            title={Return-volatility comparison},
            xmajorgrids=true,
            ymajorgrids=true,
            xlabel={$\hat{\nu}^2_\pi$},
            x tick label style={
                /pgf/number format/fixed,
                /pgf/number format/fixed zerofill,
                /pgf/number format/precision=1,
            },
            ylabel={$\hat{J}_\pi$},
            y tick label style={
                /pgf/number format/fixed,
                /pgf/number format/fixed zerofill,
                /pgf/number format/precision=1,
            },
        ]
        \addplot[
	        BaselineStyle,
        	skip coords between index={1}{9},
            visualization depends on={value \thisrow{anchor}\as\myanchor},
        ] table[x=x, y=y, meta=meta] {./default_500bps_2cds_performance.csv};
        \addplot+[
            AgentStyle,
			only marks,
            skip coords between index={0}{1},
            visualization depends on={value \thisrow{anchor}\as\myanchor},
        ] table[x=x, y=y, meta=meta] {./default_500bps_2cds_performance.csv};
    \end{axis}
    
\end{tikzpicture}
    \caption{Each dot represents the average performance, depending on $\beta$ annotated next to each dot, of an agent over 2,000 out-of-sample episodes in terms of return and unnormalized reward volatility, for a counterparty whose $\bbP$ probability of default is relevant. In all simulations $\mathcal{M}=\{\text{1Y},\text{5Y}\}$, and CVA at inception equals $-1.40\times10^{-2}$ EUR. Model parameters are specified in table~\ref{tab:many_defaults_model_pars}.}
    \label{fig:many_defaults_2cds_pareto}
\end{figure}
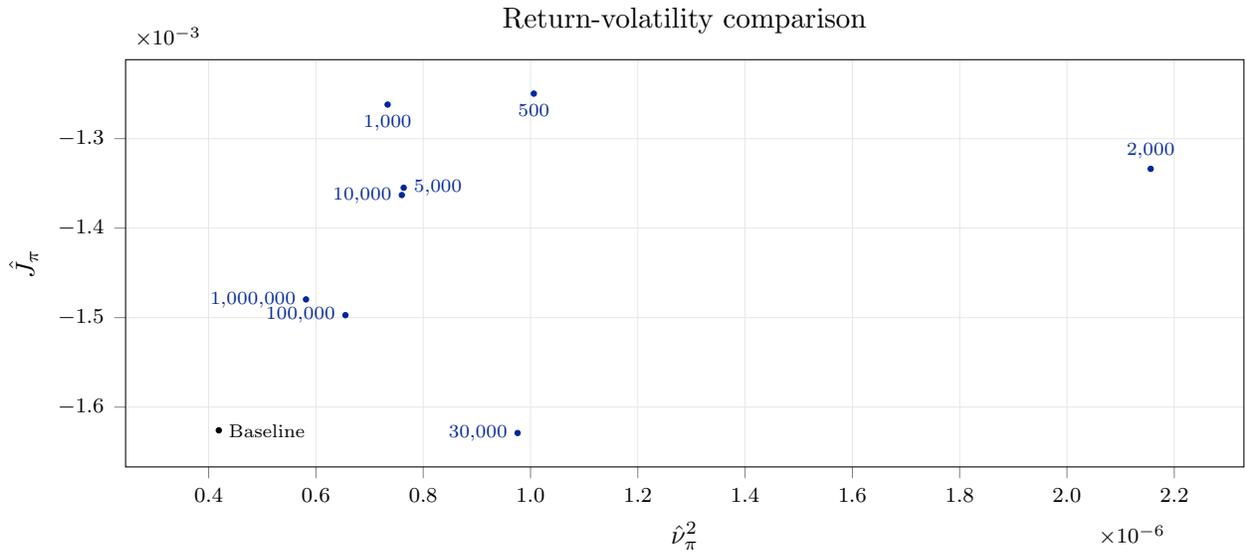
\begin{figure}
    \centering
    \begin{tikzpicture}
   	\begin{axis}
        [
            height=6cm, width=\linewidth,
            xmajorgrids=true,
            xmajorgrids=true,
            xmin=0, xmax=450,
            xlabel={Timestep},
            ylabel={Total $\lambda_t$ sensitivity},
			yticklabel style={
		        /pgf/number format/fixed,
		        /pgf/number format/fixed zerofill,
        		/pgf/number format/precision=3
			},
			scaled y ticks=false,
            legend style={
                legend columns=5,
                at={(0.5, -0.25)},
                anchor=north
            }
        ]
        \addplot[LineStyle, black] table[y=cs01_total_baseline] {./default_500bps_2cds_ep1-cr-actions.csv}; \addlegendentry{Baseline}
        \addplot[LineStyle, dashdotted, txtOrange] table[y=cs01_total_5e2] {./default_500bps_2cds_ep1-cr-actions.csv}; \addlegendentry{$\beta=500$}
        \addplot[LineStyle, dashdotted, txtRed3] table[y=cs01_total_1e3] {./default_500bps_2cds_ep1-cr-actions.csv}; \addlegendentry{$\beta=\text{1,000}$}
        \addplot[LineStyle, dashdotted, txtRed] table[y=cs01_total_2e3] {./default_500bps_2cds_ep1-cr-actions.csv}; \addlegendentry{$\beta=\text{2,000}$}
        \addplot[LineStyle, dashdotted, txtGreen] table[y=cs01_total_5e3] {./default_500bps_2cds_ep1-cr-actions.csv}; \addlegendentry{$\beta=\text{5,000}$}
        \addplot[LineStyle, dashdotted, txtGreen3] table[y=cs01_total_1e4] {./default_500bps_2cds_ep1-cr-actions.csv}; \addlegendentry{$\beta=\text{10,000}$}
        \addplot[LineStyle, dashdotted, txtBlue3] table[y=cs01_total_3e4] {./default_500bps_2cds_ep1-cr-actions.csv}; \addlegendentry{$\beta=\text{30,000}$}
        \addplot[LineStyle, dashdotted, txtBlue] table[y=cs01_total_1e5] {./default_500bps_2cds_ep1-cr-actions.csv}; \addlegendentry{$\beta=\text{100,000}$}
        \addplot[LineStyle, dashdotted, txtPurple] table[y=cs01_total_1e6] {./default_500bps_2cds_ep1-cr-actions.csv}; \addlegendentry{$\beta=\text{1,000,000}$}
    \end{axis}
    
    \begin{axis}
        [
            height=6cm, width=\linewidth,
            yshift=4.5cm,
            title={Evolution of the total CDS allocations over time},
           	xmajorgrids=true,
            xmin=0, xmax=450,
            xtick style={draw=none},
            xticklabels={,,},
            ylabel={Total notional},
			yticklabel style={
		        /pgf/number format/fixed,
		        /pgf/number format/fixed zerofill,
        		/pgf/number format/precision=3
			},
			scaled y ticks=false,
        ]
        \addplot[LineStyle, black] table[y=notional_total_baseline] {./default_500bps_2cds_ep1-cr-actions.csv};
        \addplot[LineStyle, dashdotted, txtOrange] table[y=notional_total_5e2] {./default_500bps_2cds_ep1-cr-actions.csv};
        \addplot[LineStyle, dashdotted, txtRed3] table[y=notional_total_1e3] {./default_500bps_2cds_ep1-cr-actions.csv};
        \addplot[LineStyle, dashdotted, txtRed] table[y=notional_total_2e3] {./default_500bps_2cds_ep1-cr-actions.csv};
        \addplot[LineStyle, dashdotted, txtGreen] table[y=notional_total_5e3] {./default_500bps_2cds_ep1-cr-actions.csv};
        \addplot[LineStyle, dashdotted, txtGreen3] table[y=notional_total_1e4] {./default_500bps_2cds_ep1-cr-actions.csv};
        \addplot[LineStyle, dashdotted, txtBlue3] table[y=notional_total_3e4] {./default_500bps_2cds_ep1-cr-actions.csv};
        \addplot[LineStyle, dashdotted, txtBlue] table[y=notional_total_1e5] {./default_500bps_2cds_ep1-cr-actions.csv};
        \addplot[LineStyle, dashdotted, txtPurple] table[y=notional_total_1e6] {./default_500bps_2cds_ep1-cr-actions.csv};
    \end{axis}

\end{tikzpicture}
    \caption{Total credit hedge position chosen by the baseline and by agents, both in terms of notional and in terms of sensitivity to $\smash{\lambda_t}$, for a counterparty whose $\bbP$ probability of default is relevant, for which two CDS hedges are used. Model parameters are specified in table~\ref{tab:many_defaults_model_pars}.}
    \label{fig:many_defaults_2cds_total_actions}
\end{figure}
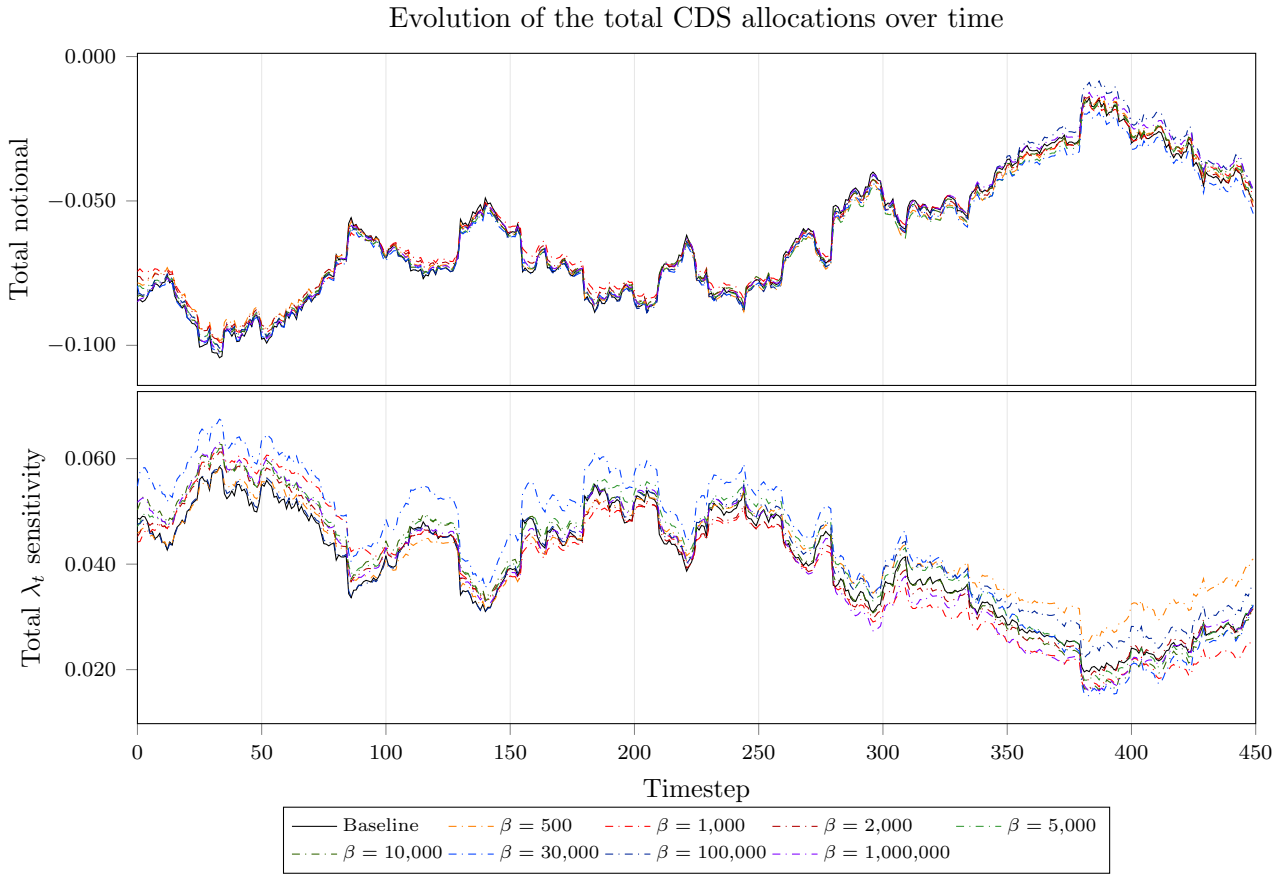
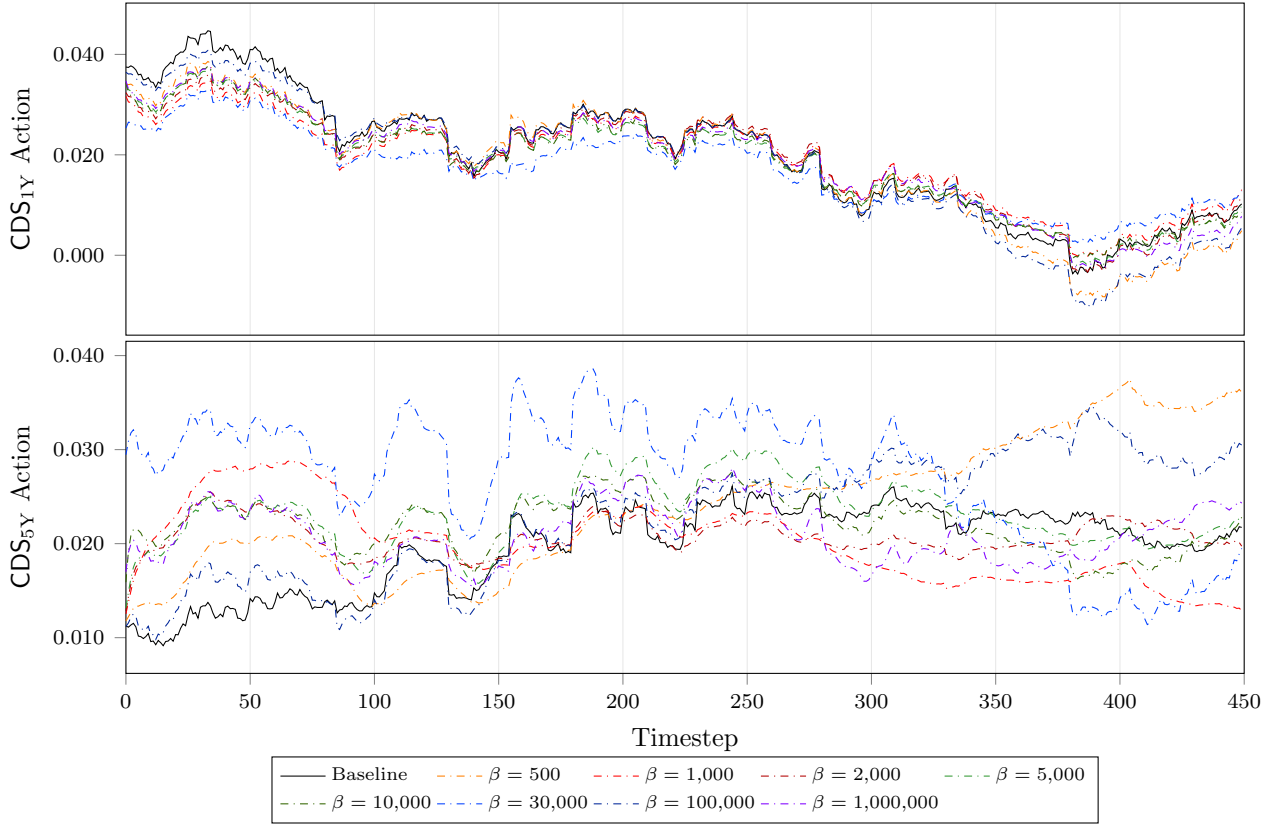
\begin{figure}
    \centering
       \begin{tikzpicture}

    \begin{axis}
        [
            height=6cm, width=\linewidth,
            xmajorgrids=true,
            xmin=0, xmax=450,
            xlabel={Timestep},
            ylabel={$\CDS_{\text{5Y}}$ Action},
			yticklabel style={
		        /pgf/number format/fixed,
		        /pgf/number format/fixed zerofill,
        		/pgf/number format/precision=3
			},
			scaled y ticks=false,
            legend style={
                legend columns=5,
                at={(0.5, -0.25)},
                anchor=north
            }
        ]
        \addplot[LineStyle, black] table[y=cs01_cds2_baseline] {./default_500bps_2cds_ep1-cr-actions.csv}; \addlegendentry{Baseline}
        \addplot[LineStyle, dashdotted, txtOrange] table[y=cs01_cds2_5e2] {./default_500bps_2cds_ep1-cr-actions.csv}; \addlegendentry{$\beta=500$}
        \addplot[LineStyle, dashdotted, txtRed3] table[y=cs01_cds2_1e3] {./default_500bps_2cds_ep1-cr-actions.csv}; \addlegendentry{$\beta=\text{1,000}$}
        \addplot[LineStyle, dashdotted, txtRed] table[y=cs01_cds2_2e3] {./default_500bps_2cds_ep1-cr-actions.csv}; \addlegendentry{$\beta=\text{2,000}$}
        \addplot[LineStyle, dashdotted, txtGreen] table[y=cs01_cds2_5e3] {./default_500bps_2cds_ep1-cr-actions.csv}; \addlegendentry{$\beta=\text{5,000}$}
        \addplot[LineStyle, dashdotted, txtGreen3] table[y=cs01_cds2_1e4] {./default_500bps_2cds_ep1-cr-actions.csv}; \addlegendentry{$\beta=\text{10,000}$}
        \addplot[LineStyle, dashdotted, txtBlue3] table[y=cs01_cds2_3e4] {./default_500bps_2cds_ep1-cr-actions.csv}; \addlegendentry{$\beta=\text{30,000}$}
        \addplot[LineStyle, dashdotted, txtBlue] table[y=cs01_cds2_1e5] {./default_500bps_2cds_ep1-cr-actions.csv}; \addlegendentry{$\beta=\text{100,000}$}
        \addplot[LineStyle, dashdotted, txtPurple] table[y=cs01_cds2_1e6] {./default_500bps_2cds_ep1-cr-actions.csv}; \addlegendentry{$\beta=\text{1,000,000}$}
    \end{axis}

    \begin{axis}
        [
            height=6cm, width=\linewidth,
            yshift=4.5cm,
            title={Evolution of the credit actions over time},
            xmajorgrids=true,
            xmin=0, xmax=450,
            xtick style={draw=none},
            xticklabels={,,},
            ylabel={$\CDS_{\text{1Y}}$ Action},
			yticklabel style={
		        /pgf/number format/fixed,
		        /pgf/number format/fixed zerofill,
        		/pgf/number format/precision=3
			},
			scaled y ticks=false,
        ]
        \addplot[LineStyle, black] table[y=cs01_cds1_baseline] {./default_500bps_2cds_ep1-cr-actions.csv};
        \addplot[LineStyle, dashdotted, txtOrange] table[y=cs01_cds1_5e2] {./default_500bps_2cds_ep1-cr-actions.csv};
        \addplot[LineStyle, dashdotted, txtRed3] table[y=cs01_cds1_1e3] {./default_500bps_2cds_ep1-cr-actions.csv};
        \addplot[LineStyle, dashdotted, txtRed] table[y=cs01_cds1_2e3] {./default_500bps_2cds_ep1-cr-actions.csv};
        \addplot[LineStyle, dashdotted, txtGreen] table[y=cs01_cds1_5e3] {./default_500bps_2cds_ep1-cr-actions.csv};
        \addplot[LineStyle, dashdotted, txtGreen3] table[y=cs01_cds1_1e4] {./default_500bps_2cds_ep1-cr-actions.csv};
        \addplot[LineStyle, dashdotted, txtBlue3] table[y=cs01_cds1_3e4] {./default_500bps_2cds_ep1-cr-actions.csv};
        \addplot[LineStyle, dashdotted, txtBlue] table[y=cs01_cds1_1e5] {./default_500bps_2cds_ep1-cr-actions.csv};
        \addplot[LineStyle, dashdotted, txtPurple] table[y=cs01_cds1_1e6] {./default_500bps_2cds_ep1-cr-actions.csv};
    \end{axis}

\end{tikzpicture}
    \caption{CDS actions chosen by the baseline and by agents on an out-of-sample episode, all expressed as the sensitivity to $\smash{\lambda_t}$ of the allocation in the hedging instrument, for a counterparty whose $\bbP$ probability of default is relevant. Model parameters are specified in table~\ref{tab:many_defaults_model_pars}.}
    \label{fig:many_defaults_2cds_splitted_actions}
\end{figure}

In this subsection we explore cases in which the counterparty is very risky, so that even common practice would acknowledge that defaults should be considered. Default probabilities are quite high, so we should be able to do without the importance sampling scheme described at the end of section~\ref{sec:simulator}, which is therefore not adopted in this section.

To focus on the effect of defaults, we assume that real-world and risk-neutral dynamics are completely equal, by setting $\bar{m}=1$, $\rho^{\bbP}_{\lambda\phi}=\rho^{\bbQ}_{\lambda\phi}$, and $\mu^{\bbP,\phi}=\mu^{\bbQ,\phi}$. Model parameters are specified in table~\ref{tab:many_defaults_model_pars}.

In this setting, if one keeps one single hedging CDS ($|\mathscr{M}|=1$), then no hedging strategy is able to cancel all risks: the allocation which cancels the sensitivity to infinitesimal movements of $\lambda$ (``delta hedge'') is not the same that ensures zero PnL if the counterparty defaults now (``jump hedge'').\footnote{This also means that the relative size of $\hat{J}_\pi$ and $\hat{\nu}^2_\pi$ changes, and so does the interesting range of values for $\beta$ where none of the two components of $\hat{\eta}_\pi$ is negligible. This is why the set of risk aversions is different here than in the previous section.} There is no reason why either should be optimal for any level of risk aversion, and we expect the optimum to be a non-trivial compromise between the two. This was already observed in the simpler setting of \citet{daluiso2023Acva}, and is confirmed with the present more realistic simulator, as one can see from figures~\ref{fig:many_defaults_1cds_pareto} and~\ref{fig:many_defaults_1cds_actions}. In the former, one can remark that many of the trained agents dominate the benchmarks along both return and volatility, although the shape of our frontier clearly indicates that learning was not perfect for a couple of the highest risk aversions. 

We then introduce a second tradable CDS with shorter maturity of 1 year, so that one can easily derive allocations that in continuous time would cancel both the diffusive and the jump risk: we use the term ``baseline'' to identify this strategy in what follows. Even compared with such a challenging benchmark the RL agent performs well. Indeed, for finite risk aversion $\beta$, it is able to achieve material cost reductions with volatilities which are of the same order of magnitude as that of the time-discretized zero-risk baseline: see figure~\ref{fig:many_defaults_2cds_pareto}. It turns out that the agent tracks the baseline's sensitivity to both $\lambda_t$ and the jump to default (figure~\ref{fig:many_defaults_2cds_total_actions}) with a non-trivial combination of the two CDS which is often quite different from the baseline (figure~\ref{fig:many_defaults_2cds_splitted_actions}).

However, in the already mentioned figure~\ref{fig:many_defaults_2cds_pareto}, the dependence on $\beta$ is unsatisfactory, so that we refrained from joining the dots with lines. Therefore, at least for some value of $\beta$, the trained agent is probably not perfectly aligned with the required risk aversion after 1,500 iterations. Importance sampling may help here, as it does in the next section in an even more challenging setup.

\subsection{Runs with rare defaults}\label{sec:run_rare_defaults}

\begin{table}
	\centering
	\begin{tabular}{llr}
		\toprule
		Description & Symbol & Value\\
		\colrule
		Instantaneous default intensity at the pricing date & $\lambda_{t_0}$ & 1.66\%\\
		Mean reversion speed of the default intensity & $k^{\bbP,\lambda}=k^{\bbQ,\lambda}$ & 0.3769\\
		Long term mean of the default intensity & $\theta^{\bbP,\lambda}=\theta^{\bbQ,\lambda}$ & 1.87\%\\
		CIR volatility coefficient of the default intensity & $\sigma^{\bbP,\lambda}=\sigma^{\bbQ,\lambda}$ &  19.22\%\\
		CDS trading cost parameter & $\gamma^{\lambda}$ & $8.3\times10^{-4}$\\
		Real-world correlation of Brownian drivers & $\rho^{\bbP}_{\lambda\phi}$ & 0\%\\
		Real-world drift of the USDEUR FX rate & $\mu^{\bbP,\phi}$ & -1.2\%\\
		\botrule
	\end{tabular}
	\caption{Value of termsheet and model parameters for section~\ref{sec:run_rare_defaults}. Credit parameters represent a counterparty with an initial flat CDS spread term structure of 100 bps with a bid-ask semi-spread of about 5 bps. Time unit is years.}\label{tab:rare_defaults_model_pars}
\end{table}

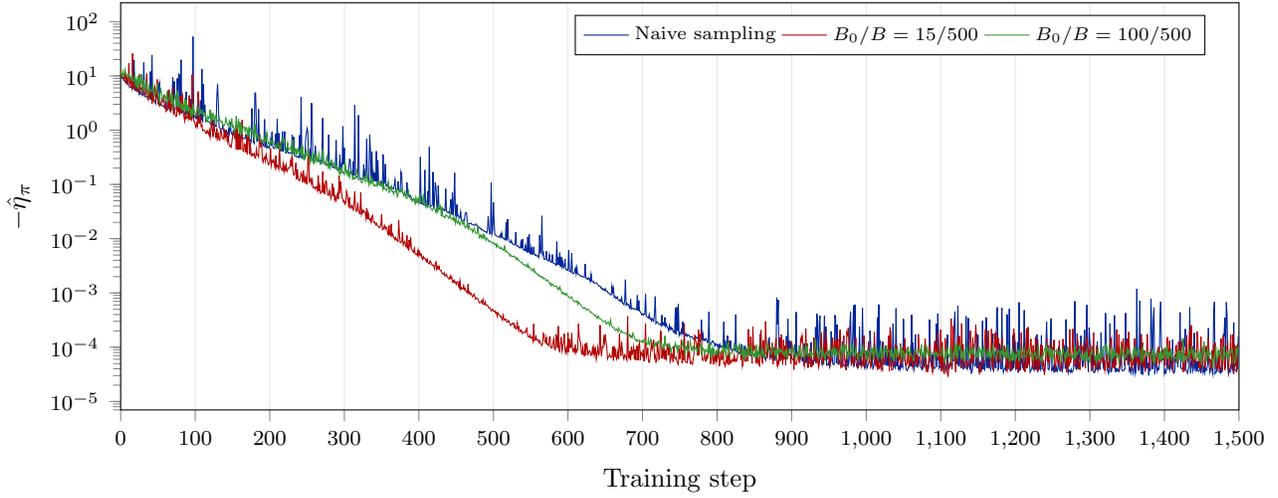
\begin{figure}
    \centering
       \begin{tikzpicture}

    \begin{axis}
        [
            height=7cm, width=\linewidth,
            title={Learning curves with one CDS hedge},
            xmajorgrids=true,
            xmin=0, xmax=1500,
            xlabel={Training step},
            ymode=log,
            log basis y={10},
            ylabel={$-\hat{\eta}_\pi$},
            legend style={
                legend columns=5,
                legend pos=north east,
            }
        ]
        
        \addplot[LineStyle, txtBlue] table[y=5e2] {./default_100bps_1cds_training.csv}; \addlegendentry{Naive sampling}		
        \addplot[LineStyle, txtRed] table[y=5e2] {./default_100bps_1cds_training_is15.csv}; \addlegendentry{$B_0/B=15/500$}
        \addplot[LineStyle, txtGreen] table[y=5e2] {./default_100bps_1cds_training_is100.csv}; \addlegendentry{$B_0/B=100/500$}
    \end{axis}

\end{tikzpicture}
    \caption{In-sample loss function as a function of the training epoch with and without importance sampling, for risk aversion parameter $\beta=500$, when only one CDS is available for hedging, and the $\bbP$ default probability is small but non-negligible.  Model parameters are specified in table~\ref{tab:rare_defaults_model_pars}. Note the log scale on the y axis.}
    \label{fig:rare_defaults_1cds_eta}
\end{figure}
\begin{figure}
    \centering
    \begin{tikzpicture}

    \begin{axis}
        [
        	table/col sep=semicolon,
            height=7cm, width=\linewidth,
            title={Return-volatility comparison},
            xmajorgrids=true,
            ymajorgrids=true,
            xlabel={$\hat{\nu}^2_\pi$},
            x tick label style={
                /pgf/number format/fixed,
                /pgf/number format/fixed zerofill,
                /pgf/number format/precision=1,
            },
            ylabel={$\hat{J}_\pi$},
            y tick label style={
                /pgf/number format/fixed,
                /pgf/number format/fixed zerofill,
                /pgf/number format/precision=1,
            },
            legend style={
                legend columns=1,
                legend pos=south east,
            }
        ]
        \addplot+[
            AgentStyle,
        	only marks,
            skip coords between index={0}{2},
            visualization depends on={value \thisrow{anchor_noIS}\as\myanchor},
        ] table[x=x_noIS, y=y_noIS, meta=meta_noIS] {./default_100bps_1cds_performance.csv};
        \addplot+[
            AgentStyle,
            txtRed,
            skip coords between index={0}{2},
            visualization depends on={value \thisrow{anchor_IS15}\as\myanchor},
        ] table[x=x_IS15, y=y_IS15, meta=meta_IS15] {./default_100bps_1cds_performance.csv};
        \addplot+[
            AgentStyle,
            txtGreen,
            skip coords between index={0}{2},
            visualization depends on={value \thisrow{anchor_IS100}\as\myanchor},
        ] table[x=x_IS100, y=y_IS100, meta=meta_IS100] {./default_100bps_1cds_performance.csv};
        \addplot+[
	        BaselineStyle,
        	skip coords between index={2}{12},
            visualization depends on={value \thisrow{anchor_noIS}\as\myanchor},
        ] table[x=x_noIS, y=y_noIS, meta=meta_noIS] {./default_100bps_1cds_performance.csv};
        \legend{Naive sampling, $B_0/B=15/500$, $B_0/B=100/500$}
    \end{axis}
    
\end{tikzpicture}
    \caption{Each dot represents the average performance depending on $\beta$, annotated next to each dot, of an agent over 20,000 out-of-sample episodes in terms of return and unnormalized reward volatility, for a counterparty whose $\bbP$ probability of default is small but non-negligible. In all simulations $\mathcal{M}=\{\text{5Y}\}$, and CVA at inception equals $-3.34\times10^{-3}$ EUR. The blue dots are built from a policy trained without importance sampling and the red and green dots from a policy trained using importance sampling with $B_0=15$ and $B_0=100$ respectively, over batches of size $B=500$. An outlier blue point for $\beta=10$ was removed for clarity. Model parameters are specified in table~\ref{tab:rare_defaults_model_pars}.}
    \label{fig:rare_defaults_1cds_pareto}
\end{figure}
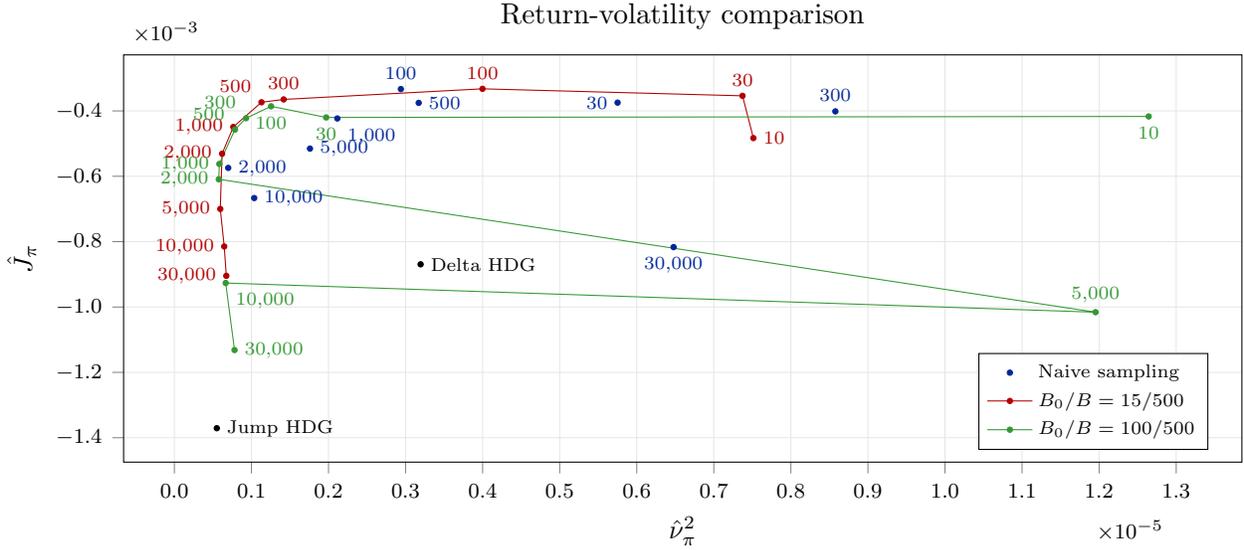
\begin{table}
	\centering
	\begin{tabular}{@{}
		c
		S[table-format=-1.3e-1, table-space-text-post={$^{***}$}]
		S[table-format=-1.3e-1, table-space-text-post={$^{***}$}]
	@{}}
		\multicolumn{3}{c}{Improvement of $\hat{\eta}_{\pi}$ with respect to naive sampling} \\
		\toprule
		{$\beta$} & {$B_0/B=15/500$} & {$B_0/B=100/500$} \\
		\colrule
		10 & 6.462E-04$^{***}$ & 6.612E-04$^{***}$ \\
		30 & -2.781E-05$^{}$ & 6.813E-05$^{*}$ \\
		100 & -1.052E-04$^{*}$ & 1.157E-04$^{***}$ \\
		300 & 2.184E-03$^{***}$ & 2.273E-03$^{***}$ \\
		500 & 1.022E-03$^{***}$ & 1.111E-03$^{***}$ \\
		1,000 & 1.324E-03$^{***}$ & 1.392E-03$^{***}$ \\
		2,000 & 2.024E-04$^{***}$ & 2.110E-04$^{***}$ \\
		5,000 & 5.636E-03$^{***}$ & -5.147E-02$^{***}$ \\
		10,000 & 3.743E-03$^{***}$ & 3.442E-03$^{***}$ \\
		30,000 & 1.740E-01$^{***}$ & 1.706E-01$^{***}$ \\
		\botrule
	\end{tabular}
	\caption{Difference between the risk-averse objective function $\hat{\eta}_\pi$ obtained by optimization with and without importance sampling, estimated on 20,000 out-of-sample episodes, for agents using one CDS hedge. The default probability of the counterparty is small but non-negligible. $^{***}$, $^{**}$, $^*$ mark significance at confidence levels 0.1\%, 1\%, and 5\% respectively. Model parameters are specified in table~\ref{tab:rare_defaults_model_pars}.}\label{tab:eta_1cds}
\end{table}

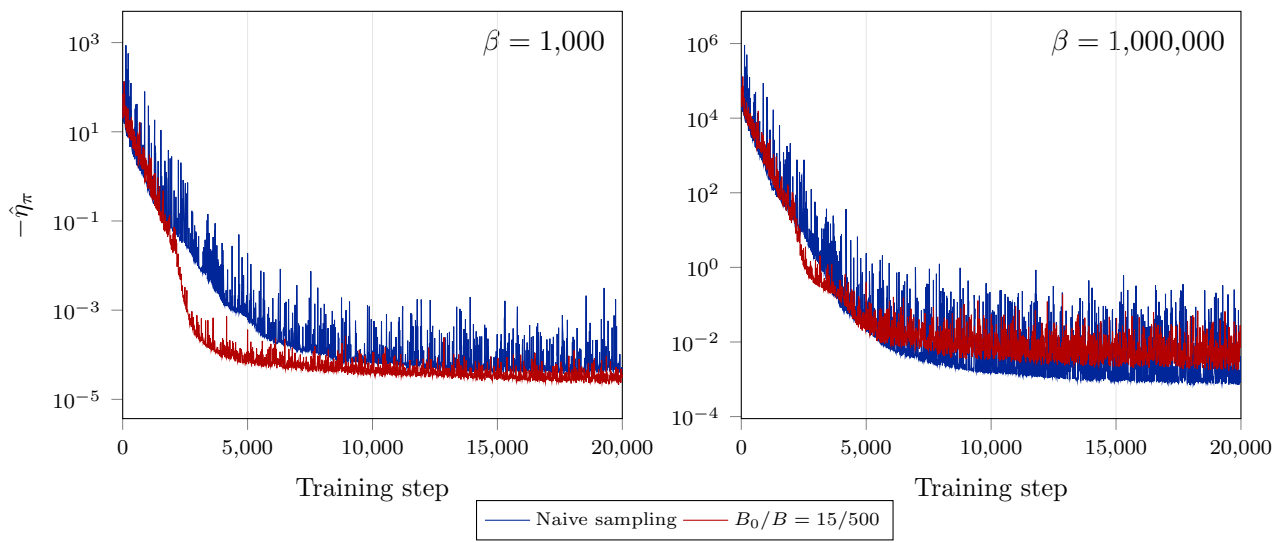
\begin{figure}
       \begin{tikzpicture}
	
	\node[align=center] at (7.45, 6){Learning curves with two CDS hedges};
	\node[] at (5.6, 5){$\beta=\text{1,000}$};
	\node[] at (13.5, 5){$\beta=\text{1,000,000}$};
    \begin{axis}
        [
            height=7cm, width=\linewidth*0.5,
            xmajorgrids=true,
            xmin=0, xmax=20000,
            xlabel={Training step},
            ymode=log,
            log basis y={10},
            scaled x ticks=false,
            ylabel={$-\hat{\eta}_\pi$},
            legend style={
                legend columns=2,
                at={(1.125, -0.2)},
                anchor=north
            }
        ]
        \addplot[LineStyle, txtBlue] table[y=1e3] {./default_100bps_2cds_training.csv}; \addlegendentry{Naive sampling}		
        \addplot[LineStyle, txtRed] table[y=1e3] {./default_100bps_2cds_training_is15.csv}; \addlegendentry{$B_0/B=15/500$}
    \end{axis}

    \begin{axis}
        [
            height=7cm, width=\linewidth*0.5,
            xshift=\linewidth*0.5,
            xmajorgrids=true,
            xmin=0, xmax=20000,
            xlabel={Training step},
            ymode=log,
            log basis y={10},
            scaled x ticks=false,
        ]
        \addplot[LineStyle, txtBlue] table[y=1e6] {./default_100bps_2cds_training.csv};
        \addplot[LineStyle, txtRed] table[y=1e6] {./default_100bps_2cds_training_is15.csv};
    \end{axis}

\end{tikzpicture}
    \caption{In-sample loss function as a function of the training epoch with and without importance sampling, for different risk aversion parameters $\beta$. The default probability of the counterparty is small but non-negligible, and two CDS hedges are available. Model parameters are specified in table~\ref{tab:rare_defaults_model_pars}. Note the log scale on the y axis.}
    \label{fig:rare_defaults_2cds_eta}
\end{figure}

\begin{figure}
    \centering
    \begin{tikzpicture}

    \begin{axis}
        [
        	table/col sep=semicolon,
            height=7cm, width=\linewidth,
            title={Return-volatility comparison},
            xmajorgrids=true,
            ymajorgrids=true,
            xmode=log,
            log basis x={10},
            xlabel={$\hat{\nu}^2_\pi$},
            ylabel={$\hat{J}_\pi$},
            y tick label style={
                /pgf/number format/fixed,
                /pgf/number format/fixed zerofill,
                /pgf/number format/precision=1,
            },
            legend style={
                legend columns=1,
                legend pos=south east,
            }
        ]
        \addplot+[
            AgentStyle,
        	only marks,
            skip coords between index={0}{1},
            visualization depends on={value \thisrow{anchor_noIS}\as\myanchor},
        ] table[x=x_noIS, y=y_noIS, meta=meta_noIS] {./default_100bps_2cds_performance.csv};
        \addplot+[
            AgentStyle,
            txtRed,
            skip coords between index={0}{1},
            visualization depends on={value \thisrow{anchor_IS15}\as\myanchor},
        ] table[x=x_IS15, y=y_IS15, meta=meta_IS15] {./default_100bps_2cds_performance.csv};
        \addplot+[
	        BaselineStyle,
        	skip coords between index={1}{9},
            visualization depends on={value \thisrow{anchor_noIS}\as\myanchor},
        ] table[x=x_noIS, y=y_noIS, meta=meta_noIS] {./default_100bps_2cds_performance.csv};
        \legend{Naive sampling, $B_0/B=15/500$}
    \end{axis}
    
\end{tikzpicture}
    \caption{Each dot represents the average performance, depending on $\beta$ annotated next to each dot, of an agent over 20,000 out-of-sample episodes in terms of return and unnormalized reward volatility, for a counterparty whose $\bbP$ probability of default is small but non-negligible. In all simulations $\mathcal{M}=\{\text{1Y},\text{5Y}\}$, $\smash{\gamma^{\lambda}=8.3\times10^{-4}}$ (i.e., bid-ask semi-spread of about 5 bps), $\smash{\gamma^{\USD}=5\times10^{-5}}$, $\smash{\rho^{\bbP}_{\lambda\phi}=0}$ and CVA at inception equals $-3.34\times10^{-3}$ EUR. The blue dots are built from a policy trained without importance sampling and the red dots from a policy trained using importance sampling with $B_0=15$ over batches of size $B=500$. An outlier blue point for $\beta=500$ was removed for clarity. Note the log scale on the x axis. Model parameters are specified in table~\ref{tab:rare_defaults_model_pars}.}
    \label{fig:rare_defaults_2cds_pareto}
\end{figure}
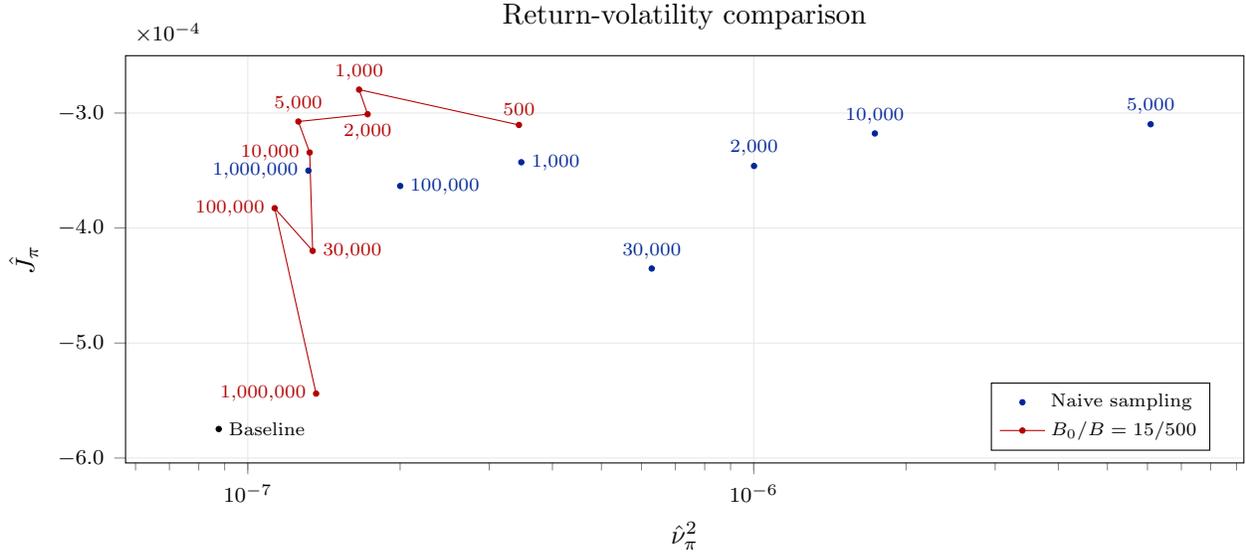
\begin{table}
	\centering
	\begin{tabular}{@{}
		c
		S[table-format=-1.3e-1, table-space-text-post={$^{***}$}]
	@{}}
		\multicolumn{2}{c}{$\hat{\eta}_\pi$ improvement with importance sampling} \\
		\toprule
		{$\beta$} & {$B_0/B=15/500$} \\
		\colrule
		500 & 1.149E+00$^{***}$ \\
		1,000 & 2.444E-04$^{***}$ \\
		2,000 & 1.702E-03$^{***}$ \\
		5,000 & 2.981E-02$^{*}$ \\
		10,000 & 1.601E-02$^{***}$ \\
		30,000 & 1.484E-02$^{***}$ \\
		100,000 & 8.689E-03$^{***}$ \\
		1,000,000 & -4.883E-03$^{}$ \\
		\botrule
	\end{tabular}
	\caption{Difference between the risk-averse objective function $\hat{\eta}_\pi$ obtained by optimization with and without importance sampling, estimated on 20,000 out-of-sample episodes, for agents using two CDS hedges. The default probability of the counterparty is small but non-negligible. $^{***}$, $^{**}$, $^*$ mark significance at confidence levels 0.1\%, 1\%, and 5\% respectively. Model parameters are specified in table~\ref{tab:rare_defaults_model_pars}.}\label{tab:eta_2cds}
\end{table}

Finally, we consider the case in which the counterparty has only a moderate initial hazard rate, but we want to overcome the over-simplistic assumption $\bar{m}\ll 1$ of section~\ref{sec:run_no_defaults}. Therefore, we assume $\bbP=\bbQ$ as in the previous section. Model parameters are specified in table~\ref{tab:rare_defaults_model_pars}. We start from the case $\mathscr{M}=\{\mathrm{5Y}\}$.

First of all, we notice that importance sampling becomes very valuable, because the rareness of the default event in naive sampling causes noisy training curves which converge slowly: see figure~\ref{fig:rare_defaults_1cds_eta}, where the learning curves with importance sampling are also showed for two different values of the $B_0$ hyperparameter. Interestingly, its impact on training is non-trivial: while a higher number of defaults per batch yields a smoother learning curve, the faster convergence is achieved with the number of defaults per batch quite low, although higher than the original probability of default within the trading horizon (which would imply approximately 3 defaults per batch). Therefore, it seems that in our experiment the benefit of importance sampling does not come only from high oversampling of the rare event, but also from making it appear the exact same number of times in each batch, a characteristic which is specific of our \textit{ad hoc} scheme as underlined in remark~\ref{rem:adaptive_importance_sampling}. This trade-off of training speed and accuracy may support an increasing schedule of values of $B_0$ along iterations: we leave this fine-tuning to future research.

The performance metrics are plotted in figure~\ref{fig:rare_defaults_1cds_pareto}. Even without importance sampling, we do get agents with an interesting risk-cost profile when compared to the two benchmarks; but while moving $\beta$, the point in return-volatility space wanders without any coherent pattern, which is a clear signal of imperfect learning. On the other hand, importance sampling with 15 defaults per batch gives a smooth curve with consistently better results. The agent's policy is less costly than the cheapest benchmark (delta hedge) for all risk aversions but the highest one, while often getting close to the volatility of the less volatile benchmark (jump hedge). Increasing the number of defaults to 100 yields comparable results, and sometimes even slightly better ones in terms of the objective function $\hat{\eta}_{\pi}$ (table~\ref{tab:eta_1cds}), but a couple of points are clearly off the efficient frontier, so probably such a large distortion of the original distribution is not advisable for robustness concerns.

As we saw that even for the low-risk counterparty under analysis, the purely-diffusive delta hedge is far from optimal, it makes sense to consider the case $\mathscr{M}=\{\mathrm{1Y},\mathrm{5Y}\}$ in which an additional hedging instrument is available. In view of the remarks made in the previous paragraph, we fix $B_0=15$.

The noise reduction effect of importance sampling in training is confirmed whichever the risk aversion, as one can see from both the low-$\beta$ left plot in and the high-$\beta$ right plot in figure~\ref{fig:rare_defaults_2cds_eta}; while the improvement in training speed seems less relevant in the extreme risk aversion setting of the latter. However, after 20,000 epochs, the training with naive sampling still completely lacks a coherent dependence on risk aversion, and its reward volatility is one or more orders of magnitude larger than the baseline's for most values of $\beta$; while importance sampling is closer to drawing a decent Pareto frontier (figure~\ref{fig:rare_defaults_2cds_pareto}). Indeed, the values of $\hat{\eta}_\pi$ are consistently better with importance sampling except for the extreme risk aversion $\beta=\text{1,000,000}$, as one can appreciate from table~\ref{tab:eta_2cds}.

\section{Conclusion}\label{sec:conclusion}

In this paper, we framed the problem of hedging Credit Valuation Adjustments under a risk-averse Reinforcement Learning perspective, with the intent to allow for any desired level of realism.

To this aim, we first of all described carefully the mechanics of the hedging book. Thanks to the adoption of model-free Reinforcement Learning, we were able to delve into the details of collateralization and transaction costs, and to highlight the role of the pricing model as opposed to the data generating one. We believe that our clear statement of the problem aligned with industry needs may have independent value for researchers.

Secondly, we highlighted the need of a custom definition of risk aversion, again starting from market practice. We found in Trust Region Volatility Optimization an algorithm very close to our needs, but for its infinite horizon formulation: therefore, we adapted it to our random finite horizon setting. This contribution may be of interest for the Reinforcement Learning community, also outside the scope of finance.

Thirdly, we tested the proposed approach in market scenarios of increasing complexity. The optimized agents are always able to save on costs by a wise reduction of hedge rebalances, but also cleverly exploit the expected returns arising from any difference between risk-neutral and real-world probabilities, whenever the risk aversion is not very high. When default events are introduced, they impact the optimal strategies in non-obvious ways, and the relevant benchmarks are regularly outdone. Even the hedging of netting sets with low default probability is significantly improved, especially if a special importance sampling scheme is adopted to avoid the noise which the rare default event would introduce in training batches. All these evidences should encourage the usage of algorithms like ours in production systems, be it for CVA or for other complex claims.

\section*{Acknowledgements}

The authors are grateful to the organizing committee of the XXIV Workshop on Quantitative Finance for the opportunity to present a working version of the present research project in April 2023.

\section*{Disclaimer}

The authors report no potential competing interest. The opinions expressed in this document are solely those of the authors and do not represent in any way those of their present and past employers.

\bibliographystyle{rQUF}
\bibliography{Bibl}

\begin{thebibliography}{67}
\providecommand{\natexlab}[1]{#1}
\providecommand{\noopsort}[1]{}
\providecommand{\printfirst}[2]{#1}
\providecommand{\singleletter}[1]{#1}
\providecommand{\switchargs}[2]{#2#1}

\bibitem[\protect\citeauthoryear{Alexander and
  Nogueira}{2007}]{alexander2007model}
Alexander, C. and Nogueira, L.M., Model-Free Hedge Ratios and Scale-Invariant
  Models. {\itshape Journal of Banking \& Finance}, 2007, \textbf{31},
  1839--1861.

\bibitem[\protect\citeauthoryear{Alexander
  {\itshape{et~al.}}}{2012}]{alexander2012regime}
Alexander, C., Rubinov, A., Kalepky, M. and Leontsinis, S., Regime-Dependent
  Smile-Adjusted Delta Hedging. {\itshape Journal of Futures Markets}, 2012,
  \textbf{32}, 203--229.

\bibitem[\protect\citeauthoryear{Andersen and
  Pykhtin}{2018}]{andersen2018margin}
Andersen, L.B.G. and Pykhtin, M. (Eds) {\itshape Margin in Derivatives
  Trading}, 2018  (Risk Books: London, UK).

\bibitem[\protect\citeauthoryear{Antonov
  {\itshape{et~al.}}}{2017}]{antonov2017pv}
Antonov, A., Issakov, S., Konikov, M., McClelland, A. and Mechkov, S., {PV} and
  {XVA} Greeks for Callable Exotics by Algorithmic Differentiation.  Available
  at SSRN, 2017.

\bibitem[\protect\citeauthoryear{Bartlett}{2006}]{bartlett2006hedging}
Bartlett, B., Hedging under {SABR} model. {\itshape Wilmott Magazine}, 2006,
  \textbf{2023}, 1--4.

\bibitem[\protect\citeauthoryear{Bellman}{1966}]{bellman1966dynamic}
Bellman, R., Dynamic programming. {\itshape Science}, 1966, \textbf{153},
  34--37.

\bibitem[\protect\citeauthoryear{Bisi
  {\itshape{et~al.}}}{2021}]{bisi2021foreign}
Bisi, L., Liotet, P., Sabbioni, L., Reho, G., Montali, N., Restelli, M. and
  Corno, C., Foreign Exchange Trading: A Risk-Averse Batch Reinforcement
  Learning Approach. In {\itshape Proceedings of the }{\itshape 1st {ACM}
  International Conference on {AI} in Finance}, {ICAIF} '20, New York, NY, USA,
  2021  (Association for Computing Machinery: New York, NY, USA).

\bibitem[\protect\citeauthoryear{Bisi {\itshape{et~al.}}}{2020}]{Bisi2020trvo}
Bisi, L., Sabbioni, L., Vittori, E., Papini, M. and Restelli, M., Risk-Averse
  Trust Region Optimization for Reward-Volatility Reduction. In {\itshape
  Proceedings of the }{\itshape 29th International Joint Conference on
  Artificial Intelligence}, edited by C.~Bessiere, {IJCAI} '20, Virtual Event,
  Jul., pp. 4583--4589, 2020  (International Joint Conferences on Artificial
  Intelligence Organization: Yokohama, Japan).

\bibitem[\protect\citeauthoryear{Bisi {\itshape{et~al.}}}{2022}]{bisi2022risk}
Bisi, L., Santambrogio, D., Sandrelli, F., Tirinzoni, A., Ziebart, B.D. and
  Restelli, M., Risk-Averse Policy Optimization via Risk-Neutral Policy
  Optimization. {\itshape Artificial Intelligence}, 2022, \textbf{311}.

\bibitem[\protect\citeauthoryear{Brigo and Mercurio}{2013}]{brigurio}
Brigo, D. and Mercurio, F., {\itshape Interest Rate Models - Theory and
  Practice}, Springer Finance, 2013  (Springer: Heidelberg, Germany).

\bibitem[\protect\citeauthoryear{Brigo
  {\itshape{et~al.}}}{2013}]{brigo2013counterparty}
Brigo, D., Morini, M. and Pallavicini, A., {\itshape Counterparty Credit Risk,
  Collateral and Funding: With Pricing Cases For All Asset Classes}, 2013
  (John Wiley \& Sons, Inc.: Hoboken, NJ, USA).

\bibitem[\protect\citeauthoryear{Buehler
  {\itshape{et~al.}}}{2019}]{buehler2019deep}
Buehler, H., Gonon, L., Teichmann, J. and Wood, B., Deep Hedging. {\itshape
  Quantitative Finance}, 2019, \textbf{19}, 1271--1291.

\bibitem[\protect\citeauthoryear{Burnett}{2021}]{burnett2021hedging}
Burnett, B., Hedging valuation adjustment: fact and friction. {\itshape Risk
  Magazine}, 2021.

\bibitem[\protect\citeauthoryear{Cao {\itshape{et~al.}}}{2023}]{cao2023gamma}
Cao, J., Chen, J., Farghadani, S., Hull, J., Poulos, Z., Wang, Z. and Yuan, J.,
  Gamma and Vega Hedging using Deep Distributional Reinforcement Learning.
  {\itshape Frontiers in Artificial Intelligence}, 2023, \textbf{6}, 1--11.

\bibitem[\protect\citeauthoryear{Cao {\itshape{et~al.}}}{2021}]{cao2021deep}
Cao, J., Chen, J., Hull, J. and Poulos, Z., Deep Hedging of Derivatives Using
  Reinforcement Learning. {\itshape Journal of Financial Data Science}, 2021,
  \textbf{3}, 10--27.

\bibitem[\protect\citeauthoryear{Cesari
  {\itshape{et~al.}}}{2009}]{cesari2009modelling}
Cesari, G., Aquilina, J., Charpillon, N., Filipović, Z., Lee, G.,  and Manda,
  I., {\itshape Modelling, Pricing, and Hedging Counterparty Credit Exposure},
  Springer Finance, 2009  (Springer: Heidelberg, Germany).

\bibitem[\protect\citeauthoryear{Chow {\itshape{et~al.}}}{2017}]{chow2017risk}
Chow, Y., Ghavamzadeh, M., Janson, L. and Pavone, M., Risk-Constrained
  Reinforcement Learning with Percentile Risk Criteria. {\itshape Journal of
  Machine Learning Research}, 2017, \textbf{18}, 6070–--6120.

\bibitem[\protect\citeauthoryear{Cox {\itshape{et~al.}}}{1985}]{cox1985theory}
Cox, J.C., Ingersoll, J.E. and Ross, S.A., A Theory of the Term Structure of
  Interest Rates. {\itshape Econometrica}, 1985, \textbf{53}, 385--407.

\bibitem[\protect\citeauthoryear{Cr{\'{e}}pey}{2004}]{crepey2004delta}
Cr{\'{e}}pey, S., Delta-Hedging Vega Risk?. {\itshape Quantitative Finance},
  2004, \textbf{4}, 559--579.

\bibitem[\protect\citeauthoryear{Crotti}{2016}]{crotti2016reinforcement}
Crotti, M.G., Reinforcement Learning for the Debt Value Adjustment Hedging.
  Master's thesis, Politecnico di Milano, 2015-2016.

\bibitem[\protect\citeauthoryear{Daluiso}{2023}]{daluiso2023fast}
Daluiso, R., Fast and Stable Credit Gamma of {CVA}. {\itshape arXiv preprint
  arXiv:2311.11672}, 2023.

\bibitem[\protect\citeauthoryear{Daluiso and Morini}{2017}]{daluiso2017hedging}
Daluiso, R. and Morini, M., Hedging Efficiently under Correlation. {\itshape
  Quantitative Finance}, 2017, \textbf{17}, 1535--1547.

\bibitem[\protect\citeauthoryear{Daluiso
  {\itshape{et~al.}}}{2023}]{daluiso2023Acva}
Daluiso, R., Pinciroli, M., Trapletti, M. and Vittori, E., {CVA} Hedging with
  Reinforcement Learning. In {\itshape Proceedings of the }{\itshape 4th {ACM}
  International Conference on {AI} in Finance}, {ICAIF} '23, New York, NY, USA,
  Nov., pp. 261--269, 2023  (Association for Computing Machinery: New York, NY,
  USA).

\bibitem[\protect\citeauthoryear{Davis
  {\itshape{et~al.}}}{1993}]{davis1993european}
Davis, M.H.A., Panas, V.G. and Zariphopoulou, T., European Option Pricing with
  Transaction Costs. {\itshape {SIAM} Journal on Control and Optimization},
  1993, \textbf{31}, 470--493.

\bibitem[\protect\citeauthoryear{Deelstra
  {\itshape{et~al.}}}{2022}]{deelstra2022accelerated}
Deelstra, G., Grzelak, L.A. and Wolf, F., Accelerated Computations of
  Sensitivities for x{VA}. {\itshape arXiv preprint arXiv:2211.17026}, 2022.

\bibitem[\protect\citeauthoryear{Dewynne
  {\itshape{et~al.}}}{1994}]{dewynne1994path}
Dewynne, J.N., Whalley, A.E. and Wilmott, P., Path-Dependent Options and
  Transaction Costs. {\itshape Philosophical Transactions: Physical Sciences
  and Engineering}, 1994, \textbf{347}, 517--529.

\bibitem[\protect\citeauthoryear{Du {\itshape{et~al.}}}{2020}]{du2020deep}
Du, J., Jin, M., Kolm, P.N., Ritter, G., Wang, Y. and Zhang, B., Deep
  Reinforcement Learning for Option Replication and Hedging. {\itshape Journal
  of Financial Data Science}, 2020, \textbf{2}, 44--57.

\bibitem[\protect\citeauthoryear{Frank
  {\itshape{et~al.}}}{2008}]{frank2008rare}
Frank, J., Mannor, S. and Precup, D., Reinforcement Learning in the Presence of
  Rare Events. In {\itshape Proceedings of the }{\itshape 25th International
  Conference on Machine Learning}, {ICML} '08, Helsinki, Finland, Jan., pp.
  336--343, 2008  (Association for Computing Machinery: New York, NY, USA).

\bibitem[\protect\citeauthoryear{Garc{{\'i}}a and
  Fern{{\'a}}ndez}{2015}]{garcia15a}
Garc{{\'i}}a, J. and Fern{{\'a}}ndez, F., A Comprehensive Survey on Safe
  Reinforcement Learning. {\itshape {JMLR}}, 2015, \textbf{16}, 1437--1480.

\bibitem[\protect\citeauthoryear{Giles
  {\itshape{et~al.}}}{2023}]{giles2023efficient}
Giles, M.B., Haji-Ali, A.L. and Spence, J., Efficient Risk Estimation for the
  Credit Valuation Adjustment. {\itshape arXiv preprint arXiv:301.05886}, 2023.

\bibitem[\protect\citeauthoryear{Gnoatto
  {\itshape{et~al.}}}{2021}]{gnoatto2021deep}
Gnoatto, A., Picarelli, A. and Reisinger, C., Deep {XVA} Solver - A Neural
  Network Based Counterparty Credit Risk Management Framework. {\itshape Risk
  Magazine}, 2021.

\bibitem[\protect\citeauthoryear{Halperin}{2019}]{halperin2019qlbs}
Halperin, I., The {QLBS} {Q}-Learner Goes {N}u{QL}ear: Fitted {Q} Iteration,
  Inverse {RL}, and Option Portfolios. {\itshape Quantitative Finance}, 2019,
  \textbf{19}, 1543--1553.

\bibitem[\protect\citeauthoryear{Halperin}{2020}]{halperin2020qlbs}
Halperin, I., {QLBS}: {Q}-Learner in the {B}lack-{S}choles(-{M}erton) Worlds.
  {\itshape The Journal of Derivatives}, 2020, \textbf{28}, 99--122.

\bibitem[\protect\citeauthoryear{Hodges and
  Neuberger}{1989}]{hodges1989optimal}
Hodges, S.D. and Neuberger, A., Optimal Replication of Contingent Claims under
  Transaction Costs. {\itshape The Review of Futures Markets}, 1989,
  \textbf{8}, 222--242.

\bibitem[\protect\citeauthoryear{Hull and White}{2017}]{hull2017optimal}
Hull, J. and White, A., Optimal Delta Hedging for Options. {\itshape Journal of
  Banking \& Finance}, 2017, \textbf{82}, 180--190.

\bibitem[\protect\citeauthoryear{Hutchinson
  {\itshape{et~al.}}}{1994}]{hutchinson1994nonparametric}
Hutchinson, J.M., Lo, A.W. and Poggio, T., A Nonparametric Approach to Pricing
  and Hedging Derivative Securities Via Learning Networks. {\itshape The
  Journal of Finance}, 1994, \textbf{49}, 851--889.

\bibitem[\protect\citeauthoryear{{International Swaps and Derivatives
  Association}}{2021}]{isda2021cds}
{International Swaps and Derivatives Association}, {CDS} Standard Model
  Documentation. \texttt{https://www.cdsmodel.com/documentation.html?}, 2021.

\bibitem[\protect\citeauthoryear{Kallsen}{1999}]{kallsen1999utility}
Kallsen, J., A Utility Maximization Approach to Hedging in Incomplete Markets.
  {\itshape Mathematical Methods of Operations Research}, 1999, \textbf{50},
  321--–338.

\bibitem[\protect\citeauthoryear{Kolm and Ritter}{2019}]{kolm2019dynamic}
Kolm, P.N. and Ritter, G., Dynamic Replication and Hedging: A Reinforcement
  Learning Approach. {\itshape Journal of Financial Data Science}, 2019,
  \textbf{1}, 159--171.

\bibitem[\protect\citeauthoryear{Leland}{1985}]{leland1985option}
Leland, H.E., Option Pricing and Replication with Transactions Costs. {\itshape
  The Journal of Finance}, 1985, \textbf{40}, 1283--1301.

\bibitem[\protect\citeauthoryear{Locatelli}{2021}]{locatelli2021two}
Locatelli, L., Two Reinforcement Learning Algorithms for Trading and {DVA}
  Hedging Problems. Master's thesis, Politecnico di Milano, 2020-2021.

\bibitem[\protect\citeauthoryear{Malliaris and
  Salchenberger}{1993}]{malliaris1993neural}
Malliaris, M. and Salchenberger, L., A Neural Network Model for Estimating
  Option Prices. {\itshape Journal of Applied Intelligence}, 1993, \textbf{3},
  193--206.

\bibitem[\protect\citeauthoryear{Mandelli
  {\itshape{et~al.}}}{2023}]{mandelli2023hedging}
Mandelli, F., Pinciroli, M., Trapletti, M. and Vittori, E., Reinforcement
  Learning for Credit Index Option Hedging. {\itshape arXiv preprint
  arXiv:2307.09844}, 2023.

\bibitem[\protect\citeauthoryear{Mikkil{\"{a}} and
  Kanniainen}{2023}]{mikkila2023empirical}
Mikkil{\"{a}}, O. and Kanniainen, J., Empirical Deep Hedging. {\itshape
  Quantitative Finance}, 2023, \textbf{23}, 111--122.

\bibitem[\protect\citeauthoryear{Moody and Saffell}{2001}]{moody2001learning}
Moody, J. and Saffell, M., Learning to Trade via Direct Reinforcement.
  {\itshape {IEEE} Transactions on Neural Networks}, 2001, \textbf{12},
  875--889.

\bibitem[\protect\citeauthoryear{Morimura
  {\itshape{et~al.}}}{2010}]{morimura2010nonparametric}
Morimura, T., Sugiyama, M., Kashima, H., Hachiya, H. and Tanaka, T.,
  Nonparametric Return Distribution Approximation for Reinforcement Learning.
  In {\itshape Proceedings of the }{\itshape 27th International Conference on
  Machine Learning}, {ICML} '10, Haifa, Israel, pp. 799--806, 2010  (Omnipress:
  Madison, WI, USA).

\bibitem[\protect\citeauthoryear{Murray
  {\itshape{et~al.}}}{2022}]{murray2022deep}
Murray, P., Wood, B., Buehler, H., Wiese, M. and Pakkanen, M., Deep Hedging:
  Continuous Reinforcement Learning for Hedging of General Portfolios across
  Multiple Risk Aversions. In {\itshape Proceedings of the }{\itshape 3rd {ACM}
  International Conference on {AI} in Finance}, {ICAIF} '22, New York, NY, USA,
  11, pp. 361--368, 2022  (Association for Computing Machinery: New York, NY,
  USA).

\bibitem[\protect\citeauthoryear{Palmisano}{2019}]{palmisano2019risk}
Palmisano, A., Risk-Sensitive Reinforcement Learning for the {DVA} Hedging.
  Master's thesis, Politecnico di Milano, 2018-2019.

\bibitem[\protect\citeauthoryear{Prashanth and
  Ghavamzadeh}{2013}]{prashanth_actor-critic_2014}
Prashanth, L. and Ghavamzadeh, M., Actor-Critic Algorithms for Risk-Sensitive
  {MDP}s. In {\itshape Proceedings of the }{\itshape 26th International
  Conference on Neural Information Processing Systems}, Vol. ~1 of {\itshape
  {NIPS} '13}, Lake Tahoe, Nevada, 12, pp. 252–--260, 2013  (Curran
  Associates, Inc.: Red Hook, NY, USA).

\bibitem[\protect\citeauthoryear{Puterman}{1994}]{puterman2014markov}
Puterman, M.L., {\itshape Markov Decision Processes: Discrete Stochastic
  Dynamic Programming}, {Wiley Series in Probability and Statistics}, 1994
  (John Wiley \& Sons, Inc.: New York, NY, USA).

\bibitem[\protect\citeauthoryear{Reghai
  {\itshape{et~al.}}}{2015}]{reghai2015cva}
Reghai, A., Kettani, O. and Messaoud, M., {CVA} with Greeks and {AAD}.
  {\itshape Risk Magazine}, 2015.

\bibitem[\protect\citeauthoryear{Ruf and Wang}{2020}]{ruf2020neural}
Ruf, J. and Wang, W., Neural Networks for Option Pricing and Hedging: a
  Literature Review. {\itshape Journal of Computational Finance}, 2020,
  \textbf{24}, 1--46.

\bibitem[\protect\citeauthoryear{Ruf and Wang}{2022}]{ruf2022hedging}
Ruf, J. and Wang, W., Hedging with Linear Regressions and Neural Networks.
  {\itshape Journal of Business and Economic Statistics}, 2022, \textbf{40},
  1442--1454.

\bibitem[\protect\citeauthoryear{Schulman
  {\itshape{et~al.}}}{2015}]{schulman2015trust}
Schulman, J., Levine, S., Moritz, P., Jordan, M. and Abbeel, P., Trust Region
  Policy Optimization. In {\itshape Proceedings of the }{\itshape 32nd
  International Conference on Machine Learning}, Vol. ~37 of {\itshape {ICML}
  '15}, Jul., pp. 1889--1897, 2015  (PMLR: Lille, France).

\bibitem[\protect\citeauthoryear{Shen
  {\itshape{et~al.}}}{2014}]{shen_risk-averse_2014}
Shen, Y., Huang, R., Yan, C. and Obermayer, K., Risk-Averse Reinforcement
  Learning for Algorithmic Trading. In {\itshape Proceedings of the }{\itshape
  2014 {IEEE} Conference on Computational Intelligence for Financial
  Engineering \& Economics}, {CIFEr}, London, UK, Mar., pp. 391--398, 2014
  (IEEE: Piscataway, NJ, USA).

\bibitem[\protect\citeauthoryear{Sutton}{1988}]{sutton1988learning}
Sutton, R.S., Learning to Predict by the Methods of Temporal Differences.
  {\itshape Machine learning}, 1988, \textbf{3}, 9--44.

\bibitem[\protect\citeauthoryear{Tamar
  {\itshape{et~al.}}}{2015}]{tamar2015policy}
Tamar, A., Chow, Y., Ghavamzadeh, M. and Mannor, S., Policy Gradient for
  Coherent Risk Measures. In {\itshape Proceedings of the }{\itshape 28th
  International Conference on Neural Information Processing Systems}, Vol. ~1
  of {\itshape {NIPS} '15}, Montreal, Canada, pp. 1468--1476, 2015  ({MIT}
  Press: Cambridge, MA, USA).

\bibitem[\protect\citeauthoryear{Tamar
  {\itshape{et~al.}}}{2017}]{tamar_sequential_2017}
Tamar, A., Chow, Y., Ghavamzadeh, M. and Mannor, S., Sequential Decision Making
  with Coherent Risk. {\itshape {IEEE} Transactions on Automatic Control},
  2017, \textbf{62}, 3323--3338.

\bibitem[\protect\citeauthoryear{Tamar
  {\itshape{et~al.}}}{2012}]{tamar2012policy}
Tamar, A., Di~Castro, D. and Mannor, S., Policy Gradients with Variance Related
  Risk Criteria. In {\itshape Proceedings of the }{\itshape 29th International
  Coference on International Conference on Machine Learning}, {ICML} '12,
  Edinburgh, Scotland, pp. 1651--1658, 2012  (Omnipress: Madison, WI, USA).

\bibitem[\protect\citeauthoryear{Tamar and Mannor}{2013}]{tamar2013variance}
Tamar, A. and Mannor, S., Variance Adjusted Actor Critic Algorithms. {\itshape
  arXiv preprint arXiv:1310.3697}, 2013.

\bibitem[\protect\citeauthoryear{Tizzano}{2018}]{tizzano2018direct}
Tizzano, A., Direct Reinforcement Learning for the {DVA} Hedging through
  Recurrent Generative Adversarial Networks for Dataset Augmentation. Master's
  thesis, Politecnico di Milano, 2017-2018.

\bibitem[\protect\citeauthoryear{V{\"{a}}h{\"{a}}maa}{2004}]{vahamaa2004delta}
V{\"{a}}h{\"{a}}maa, S., Delta hedging with the Smile. {\itshape Financial
  Markets and Portfolio Management}, 2004, \textbf{18}, 241--–255.

\bibitem[\protect\citeauthoryear{Vit}{2017}]{vit2018reinforcement}
Vit, G., Reinforcement Learning for {DVA} Hedging. Master's thesis, Politecnico
  di Milano, 2016-2017.

\bibitem[\protect\citeauthoryear{Vittori
  {\itshape{et~al.}}}{2020}]{vittori2020option}
Vittori, E., Trapletti, M. and Restelli, M., Option Hedging with Risk Averse
  Reinforcement Learning. In {\itshape Proceedings of the }{\itshape 1st {ACM}
  International Conference on {AI} in Finance}, {ICAIF} '20, New York, NY, USA,
  Oct., 2020  ({Association for Computing Machinery}: New York, NY, USA).

\bibitem[\protect\citeauthoryear{Whalley and
  Wilmott}{1997}]{whalley1997asymptotic}
Whalley, A.E. and Wilmott, P., An Asymptotic Analysis of an Optimal Hedging
  Model for Option Pricing with Transaction Costs. {\itshape Mathematical
  Finance}, 1997, \textbf{7}, 307--324.

\bibitem[\protect\citeauthoryear{Zakamouline}{2005}]{zakamouline2005optimal}
Zakamouline, V.I., Optimal Hedging of Options with Transaction Costs. {\itshape
  Willmott Magazine}, 2005, \textbf{2005}, 70--82.

\bibitem[\protect\citeauthoryear{Zhang
  {\itshape{et~al.}}}{2021}]{zhang2020mean}
Zhang, S., Liu, B. and Whiteson, S., Mean-Variance Policy Iteration for
  Risk-Averse Reinforcement Learning. In {\itshape Proceedings of the
  }{\itshape 35th {AAAI} Conference on Artificial Intelligence}, Vol. ~12 of
  {\itshape {AAAI-21}}, Virtual Event, 2021  (Curran Associates, Inc.: Red
  Hook, NY, USA).

\end{thebibliography}

\end{document}